\documentclass[lettersize,journal]{IEEEtran}

\usepackage{amsmath,amssymb,mathtools}
\usepackage{bm}
\usepackage{bbold}
\usepackage{enumitem}
\usepackage{booktabs}
\usepackage{hyperref}
\usepackage{cleveref}
\usepackage{algorithm}
\usepackage{algorithmicx}
\usepackage{algpseudocode}
\usepackage{amsthm}
\usepackage[table]{xcolor}

\newtheorem{remark}{Remark}
\newtheorem{theorem}{Theorem}
\newtheorem{corollary}{Corollary}
\newtheorem{proposition}{Proposition}

\begin{document}

\title{Joint Optimization of Qubit Leasing and Quantum Circuit Distribution}

\author{Anoushka Dey, Gaurav S. Kasbekar,~\IEEEmembership{Senior Member,~IEEE}
\thanks{The authors are with the Department of Electrical Engineering, Indian Institute of Technology (IIT) Bombay, Mumbai, 400076, India. Their email addresses are anoushkadey12@gmail.com and gskasbekar@ee.iitb.ac.in.}}
\maketitle

\begin{abstract}
We consider an agent, who would like to execute a given quantum circuit using resources leased from a set of quantum computers (QCs) connected by a quantum network. For this purpose, the agent needs to make the following four key decisions: (i) how many qubits to lease from each  QC, (ii) at which QCs to store different circuit qubits in different time slots, (iii) at which QC to execute each gate in the circuit, and (iv) when to teleport qubits between different pairs of QCs. We refer to this problem facing the agent as the joint qubit leasing and quantum circuit distribution (JQLQCD) problem, and provide a comprehensive integer linear programming (ILP) formulation for it.  We show that the JQLQCD problem is NP-complete. Next, we identify several special cases in which the problem can be optimally solved in closed form or via polynomial-time algorithms. Also, we propose a greedy algorithm with local search refinement to solve large instances of the general JQLQCD problem. Finally, using extensive numerical computations, we show that the proposed greedy algorithm achieves solutions within $15 \, \%$ of simulated annealing on general problem instances, while being $8-40\times$ faster, and that it attains the optimal solution in two of the special cases for which polynomial-time optimal algorithms are available.
\end{abstract}


\section{Introduction}

Quantum computing promises exponential speedups for certain computational problems, but current quantum computers (QCs) face severe limitations in qubit count and coherence time \cite{caleffi2024distributed}. Distributing quantum computations across multiple QCs connected by a quantum network offers a path to overcoming the constraints of a single device and scaling up operations \cite{barral2025review,larasati2025towards,Caleffi2022DistributedQC}.
However, this distribution introduces significant challenges in resource allocation, qubit placement, and communication management \cite{caleffi2024distributed}.

The problem of mapping quantum circuits to hardware has been extensively studied for single QCs \cite{Siraichi2018QubitAF, Zulehner2018EfficientMP, Murali2019NoiseAdaptiveCF, 
Tannu2019NotAT, Murali2020SoftwareMF, Wille2023MQTQMAP, Peham2023OptimalSubarch,Wille2016OptimalQCM,guo2024smt,Yan2024QuantumCircuitSurvey,ruan2025powermove,letras2025towards,zhu2025quantum,yang2025qubit,stade2025routing,quetschlich2025mqt,molavi2026generating,swierkowska2024achieving,ghlib2026scalable}, where the primary challenge is mapping logical qubits to physical qubits subject to connectivity constraints and gate fidelity considerations.
As quantum devices grow and quantum networks emerge, \emph{distributed quantum computing} 
has gained attention as a promising approach to overcoming the limitations of 
individual devices~\cite{caleffi2024distributed,barral2025review,larasati2025towards,Caleffi2022DistributedQC}. Various aspects of distributed quantum computing, including compilation \cite{Ferrari2021CompilerSO,ferrari2023modular,liu2025ecdqc,promponas2025compiler}, qubit allocation and circuit optimization \cite{mao2023qubit,sunkel2025time,sunkel2025evolutionary}, distributed quantum computing architectures \cite{Caleffi2022DistributedQC, Cuomo2020TowardsAD}, and circuit partitioning strategies \cite{g2021efficient,sundaram2022distribution,Andres-Martinez2019AutomatedDO, Daei2020OptimizedQC,pastor2024circuit,burt2024generalised,burt2026multilevel,wu2026efficient,kaur2025optimized,burt2025entanglement,yang2026distributing,yang2025efficient} have been studied in prior work. Also, quantum communication, entanglement distribution, resource allocation in quantum networks, and quantum cloud computing have been investigated in several papers \cite{Bennett1993TeleportingAU, Briegel1998QuantumRP, gu2024fendi, halder2024optimal, Inesta2023EntanglementDistribution, gu2025cost,fan2025distribution,fan2025optimized, Cicconetti2023ResourceAllocation, LaRose2019OverviewPQ,nguyen2024quantum}.

As the number and scale of commercial providers of resources such as qubits and gates for distributed quantum computing are expected to increase rapidly in the near future \cite{caleffi2024distributed}, an agent who would like to execute a given quantum circuit in a distributed manner must be able to efficiently determine the numbers of resources to lease out from different providers, in addition to solving the traditional problem of partitioning the circuit across the available QCs. Also, the   numbers of resources to lease from different QCs and the partitioning of the circuit must be \emph{jointly optimized} to effectively achieve various objectives, including minimization of the leasing cost, communication cost, and makespan (circuit completion time). Although the problem of quantum circuit partitioning has been extensively studied in prior work \cite{g2021efficient,sundaram2022distribution,Andres-Martinez2019AutomatedDO, Daei2020OptimizedQC,pastor2024circuit,burt2024generalised,burt2026multilevel,wu2026efficient,kaur2025optimized,burt2025entanglement,yang2026distributing,yang2025efficient}, to the best of our knowledge, joint optimization of resource leasing and quantum circuit distribution has not been addressed. This is the space in which we contribute in this paper.              


We consider an agent, who would like to execute a quantum circuit with a given set of logical qubits. Also, there is a set of QCs connected by a quantum network. Each QC from this set offers resources for lease, subject to its storage capacity, which  represents the maximum number of qubits that can be stored at the QC, and its execution capacity, which  represents the maximum number of qubits that can be actively processed during a time slot. The agent needs to execute its quantum circuit using resources leased from different QCs. For this purpose, the agent makes the following four key decisions: (i) how many qubits to lease from each  QC, (ii) at which QCs to store different circuit qubits in different time slots, (iii) which QC executes each gate in the circuit, and (iv) when to teleport qubits between different pairs of QCs. We refer to this problem facing the agent as the \emph{joint qubit leasing and quantum circuit distribution (JQLQCD) problem}. The goal is to minimize the total system cost subject to capacity and execution constraints. This total cost balances multiple components: the leasing costs for storage and execution capacity at different QCs, the gate execution costs, which vary across different QCs, the communication costs for qubit teleportation, and the circuit makespan. A challenge is that different QCs may have heterogeneous capabilities, costs, and connectivity. E.g., some QCs may offer low leasing costs, but limited gate support or costly qubit teleportations, while others may have higher leasing costs, but provide faster execution or better connectivity to other QCs. Some additional challenges are to ensure that circuit precedence constraints and QC capacity limits are satisfied  and to determine the temporal scheduling of gate executions to minimize the overall makespan, while keeping costs low.
We present an integer linear programming (ILP) formulation of the JQLQCD problem that explicitly models leasing costs, gate execution costs, and  qubit relocation via teleportation, which uses pre-shared entanglement \cite{Bennett1993TeleportingAU}.   

The main contributions of this paper are as follows:
\begin{itemize}
    \item We provide a comprehensive ILP formulation for the JQLQCD problem, which features distributed quantum circuit execution for an agent, with explicit modeling of teleportation with distance-dependent costs. Unlike prior work in which QCs autonomously negotiate or operate under federated control \cite{Andres-Martinez2019AutomatedDO}, our model features passive QCs that provide resources to a single decision-making agent. This architecture reflects emerging quantum cloud platforms where users rent resources from multiple providers~\cite{LaRose2019OverviewPQ}.
    \item We show that the JQLQCD problem is NP-complete by proving that the multiprocessor scheduling problem with precedence constraints  (denoted $P|\text{prec}|C_{\max}$ in standard scheduling notation), which is known to be NP-complete \cite{Ullman1975, GareyJohnson1979, Lenstra1977}, is polynomial-time reducible to it. 
    \item Although the general JQLQCD problem is NP-complete, several special cases admit efficient solutions. We identify several special cases in which the problem can be solved optimally in closed form or via polynomial-time algorithms, including those in which (A) there is an unlimited capacity QC in a heterogeneous network,  (B) homogeneous QCs with zero teleportation cost, (C) chain topology with sequential gates, (D)  independent subcircuits with partitioned resources, (E) infinite resources with makespan minimization only, and (F) a tree-structured circuit with an arbitrary QC network. For case (A), we derive necessary and sufficient conditions for centralized execution on a single QC with unlimited resources to be optimal, whereas for cases (B) to (F), we provide polynomial-time algorithms for optimal solution of the JQLQCD problem.  
    \item We propose a greedy heuristic algorithm with local search refinement to solve large instances of the general JQLQCD problem.
   \item Using extensive numerical computations, we demonstrate
    that our proposed greedy algorithm achieves solutions within $15 \, \%$ of
    simulated annealing \cite{kleinberg2006algorithm} on general problem
    instances, while being $8-40\times$ faster. For the special cases with
    known optimal solutions, the greedy algorithm attains the optimal solution
    for cases (A) and (C), and is within $8-25 \, \%$ of optimal for case (B).
    These results validate the practical utility of the proposed greedy
    algorithm for real-time distributed quantum computing scenarios where fast
    decision-making is critical.
\end{itemize}

The rest of this paper is organized as follows. Section \ref{sec:related} provides a  review of related work. Section~\ref{sec:system-model} presents the system model and problem formulation. Section~\ref{sec:complexity} proves the NP-completeness of the general problem. Section~\ref{sec:special-cases} analyzes polynomial-time solvable special cases and Section \ref{SC:greedy:algorithm} presents the greedy algorithm. Section~\ref{sec:num} provides our numerical results. Section \ref{sec:conc} concludes and provides some directions for future work.

\section{Related Work}
\label{sec:related}

\subsection{Single-Device Quantum Circuit Mapping}

The problem of mapping quantum circuits to physical hardware has been extensively studied for single QCs; surveys on this topic are \cite{Yan2024QuantumCircuitSurvey,zhu2025quantum}. Early work focused on qubit allocation and routing subject to connectivity constraints~\cite{Siraichi2018QubitAF, Zulehner2018EfficientMP}. These approaches typically use SWAP gate insertion to route interactions between non-adjacent qubits in the device topology. More recent work has incorporated noise-aware compilation~\cite{Murali2019NoiseAdaptiveCF, 
Tannu2019NotAT} and crosstalk mitigation~\cite{Murali2020SoftwareMF}. Recent advances 
in efficient mapping techniques~\cite{Wille2023MQTQMAP, Peham2023OptimalSubarch} have  improved the performance for single-device scenarios.

Several optimization formulations have been proposed, including SAT-based 
approaches~\cite{Wille2016OptimalQCM}, SMT solvers~\cite{guo2024smt}, 
and heuristic methods based on A* search~\cite{Zulehner2018EfficientMP}. Recent work on optimal subarchitectures~\cite{Peham2023OptimalSubarch} and efficient mapping tools~\cite{Wille2023MQTQMAP} has advanced single-device  mapping. In \cite{swierkowska2024achieving}, a quantum circuit compiler based on a multi-objective heuristic optimization approach was proposed to achieve Pareto-optimality in the compilation. The work \cite{ghlib2026scalable} proposed a scalable multi-objective genetic algorithm for quantum circuit optimization tailored to Noisy Intermediate-Scale Quantum (NISQ) devices. In \cite{ruan2025powermove}, PowerMove, an efficient compiler for neutral atom QCs with zoned architecture, which leverages qubit movement capabilities, was proposed. The work \cite{letras2025towards} presented the design of a multi-target Multi-Level Intermediate Representation (MLIR)-based quantum compiler, which supports advanced optimizations. A routing-aware placement method  for zoned neutral atom-based quantum computing architectures was proposed in \cite{stade2025routing}. In \cite{yang2025qubit}, a unified qubit mapping and routing framework applicable to diverse quantum  instruction set architectures, was proposed. A scheme for automatically generating qubit mapping and routing compilers for evaluating quantum circuits on arbitrary quantum processors was proposed in \cite{molavi2026generating}.  A framework, called the MQT Predictor, which allows one to automatically select a suitable quantum device for a particular quantum circuit and application and provides an optimized compiler for the selected device, was proposed in \cite{quetschlich2025mqt}.

However, these techniques are limited to single-device scenarios and do not address the resource allocation and communication challenges that arise in distributed quantum computing.

\subsection{Distributed Quantum Computing}

Distributed quantum computing has emerged as a promising approach to overcome the limitations of single QCs; recent surveys on this topic are \cite{caleffi2024distributed,barral2025review,larasati2025towards}. In \cite{Ferrari2021CompilerSO,ferrari2023modular,liu2025ecdqc,promponas2025compiler}, compilers for distributed quantum computing were proposed. Qubit allocation and circuit optimization for distributed quantum computing were studied in \cite{mao2023qubit,sunkel2025time,sunkel2025evolutionary}. 
The work \cite{Caleffi2022DistributedQC, Cuomo2020TowardsAD} has explored practical aspects of distributed quantum computing architectures. 

Several quantum circuit partitioning strategies have been proposed \cite{g2021efficient,sundaram2022distribution,kaur2025optimized,burt2025entanglement,yang2026distributing}. In \cite{Andres-Martinez2019AutomatedDO}, automated methods were developed for distributing quantum circuits across multiple devices, with a focus on minimizing the number of required entangled pairs. In \cite{Daei2020OptimizedQC}, optimization techniques for partitioning circuits were proposed taking into account both communication costs and device capabilities. In \cite{pastor2024circuit}, an algorithm for circuit partitioning was proposed based on deep reinforcement learning. In \cite{burt2024generalised}, a graph-based formulation was proposed for partitioning quantum circuits, which allows the joint optimization of gate and state teleportation costs. A framework for partitioning quantum circuits, which uses multilevel techniques that coarsen hypergraphs and partition at multiple levels of granularity, was proposed in \cite{burt2026multilevel}. In \cite{wu2026efficient}, a time-aware heuristic based on beam search was proposed to solve the quantum circuit partitioning problem, which is designed to minimize communication overhead without incurring prohibitive computational time. The work \cite{yang2025efficient} studied the problem of efficiently distributing multiple quantum circuits across a shared quantum network under decoherence and network constraints. 

However, none of these papers addresses the problem of jointly leasing out qubits and distributing quantum circuits, which is studied in our paper.  

\subsection{Quantum Communication and Networks, Entanglement Distribution, and Quantum Cloud Computing}
Resource allocation in quantum networks has been studied from various perspectives. Early theoretical work established the foundations of quantum communication \cite{Bennett1993TeleportingAU} and entanglement distribution \cite{Briegel1998QuantumRP}.  High-fidelity entanglement distribution and purification in quantum networks were investigated in \cite{gu2024fendi, halder2024optimal, Inesta2023EntanglementDistribution, gu2025cost,fan2025distribution,fan2025optimized}.  The work \cite{Cicconetti2023ResourceAllocation} specifically addressed 
resource allocation in quantum networks for distributed quantum computing. Prior work on quantum cloud computing has explored scenarios in which users rent resources from multiple providers \cite{LaRose2019OverviewPQ,nguyen2024quantum}.
However, none of these papers considers the problem of jointly leasing out qubits and quantum circuit distribution that we address in this paper.

\subsection{Complexity}
The complexity of quantum circuit optimization problems has been studied in various contexts. Qubit routing on constrained topologies is known to be NP-hard~\cite{Botea2018ComplexityOQ}. Makespan minimization for quantum circuits with precedence constraints relates to the classical multiprocessor scheduling problem $P|\text{prec}|C_{\max}$, which is strongly NP-hard~\cite{Ullman1975NPCompleteS, GrahamLawlerLenstraRinnooyKan1979OptimizationAA}. However, the complexity of the JQLQCD problem that we study in this paper has not been analyzed in prior work.

\section{System Model and Problem Formulation}
\label{sec:system-model}

\begin{figure*}[t]
\centering
\includegraphics[width=\textwidth]{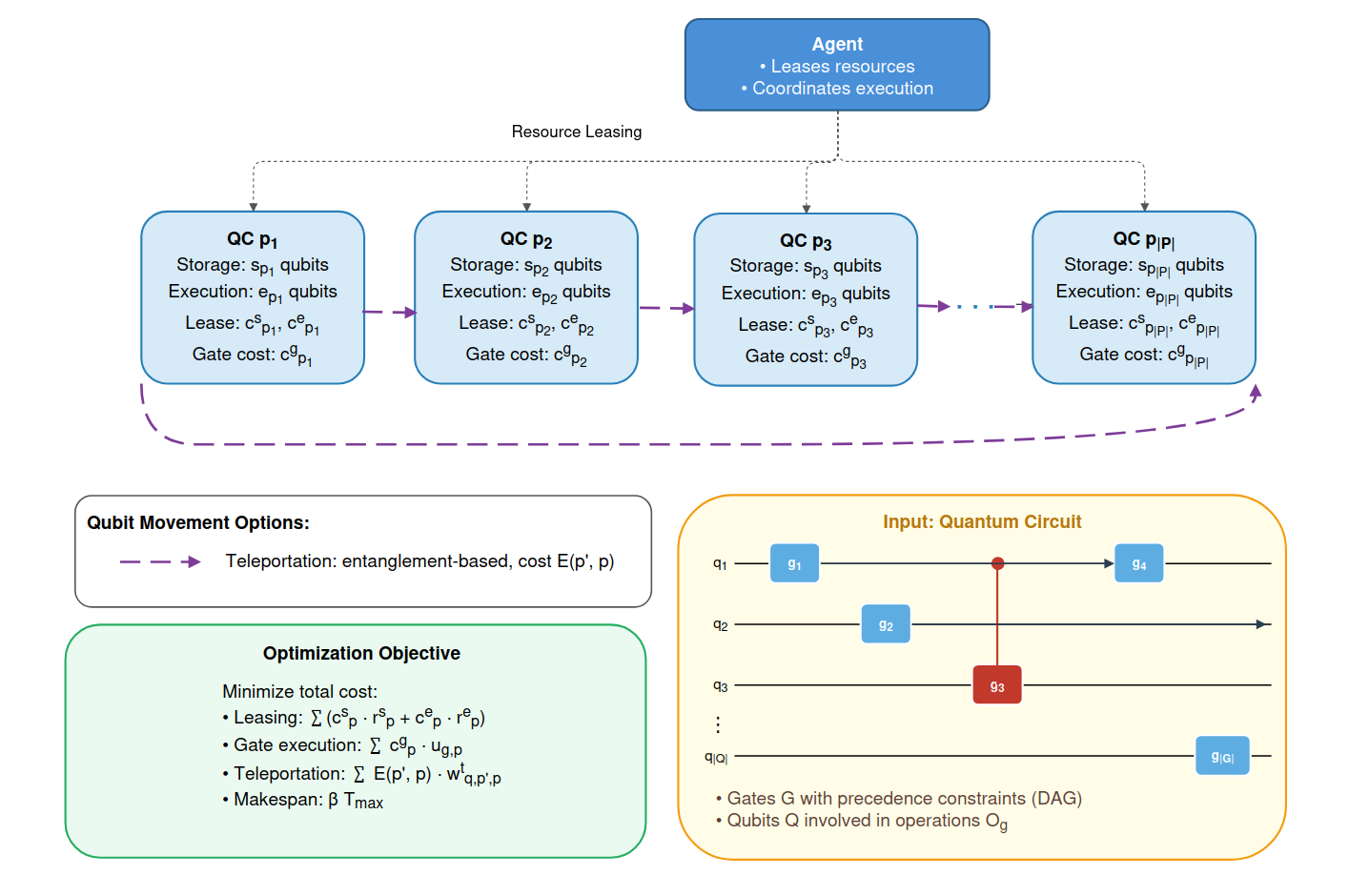}
\caption{The figure illustrates the system model. The agent leases resources 
from multiple QCs (QC $p_1$ to QC $p_{|P|}$) connected by a quantum network. 
Each QC $p$ has qubit storage capacity $s_p$ and qubit execution capacity 
$e_p$, and charges costs $c^s_p$ and $c^e_p$ for leasing storage and 
execution qubits, respectively. The execution cost of gate $g$ at QC $p$ is 
$c^g_{p}$. The figure also shows the quantum circuit that the agent would 
like to run using the leased resources.}
\label{fig:system_model}
\end{figure*}
\subsection{System Model}
Consider an agent, shown at the top of Fig. \ref{fig:system_model}, who would like to run a quantum circuit with a set of logical qubits $Q$ and a set of gates $G$. For each gate $g$, $O_g$ denotes the set of operand qubits required to execute that gate. Time is divided into discrete slots and $T$ denotes the set of slots. Also, there is a set, $P$, of QCs, which are connected by a quantum network. Each QC $p \in P$ is characterized by two capacity parameters: storage capacity $s_p$, which limits the number of qubits that can be stored simultaneously at $p$, and execution capacity $e_p$, which limits the number of qubits that can be actively processed at $p$ during any time slot. The agent needs to execute its quantum circuit using resources leased from the QCs in $P$. For this purpose, the agent needs to make four key decisions: (i) how many qubits to lease from each  QC, (ii) at which QCs to store different circuit qubits in different time slots, (iii)  which QC will execute each gate in the circuit, and (iv) when qubits will be teleported between different pairs of QCs. 

As illustrated in Fig. \ref{fig:system_model}, the available QCs, $p_1,  \ldots, p_{|P|}$, are passive resource providers connected via quantum and classical channels. The agent coordinates all resource allocation and circuit execution decisions. The bottom-right panel shows the input quantum circuit represented as a directed acyclic graph (DAG) with gates $g_1, \ldots, g_{|G|}$ that must satisfy the precedence constraints indicated by the arrows.

Table \ref{table:notation_params} describes the notation used throughout this paper. The per-qubit leasing cost  for storage (respectively, execution capacity) at QC $p \in P$ is denoted by $c^s_p$ (respectively, $c^e_p$). The  execution cost of gate $g \in G$ at QC $p \in P$ is denoted by $c^g_{p}$. Note that our model considers the general case where the different QCs in $P$ may be heterogeneous, possibly based on different storage and computation technologies, and may have different  gate fidelities, execution latencies,  etc. 

Qubit relocation using \emph{teleportation} is considered (see Fig. \ref{fig:system_model}). Teleportation transfers a quantum state using pre-shared entanglement and classical communication~\cite{Bennett1993TeleportingAU}. No physical qubit moves, but fresh entangled pairs must be consumed. The effective cost includes the number of Bell pairs that must be distributed \cite{Inesta2023EntanglementDistribution}, fidelity loss with distance, and classical communication latency.
Let \(E(p',p)\) be the cost of teleporting a qubit from  QC \(p'\) to QC \(p\); it includes the entanglement cost of sharing a Bell pair between  QCs \(p'\) and \(p\). This would typically scale with the physical distance and the  number of network hops. 
The parameter $\beta$ is the weight for the makespan (circuit completion time) term in the objective function (see \eqref{eq:obj}) and allows us to achieve different trade-offs between the makespan and the other costs. Gate availability is captured by the indicator $A_{p,g}$, which equals 1 if gate $g$ can execute at QC $p$ and 0 otherwise. This models heterogeneous QC capabilities where not all gates can be executed on all QCs. Each gate $g$ has an associated duration $\tau_g$ representing its execution time.

\begin{table}[!ht]
\caption{The table describes the notation used in this paper. }
\label{table:notation_params}
\centering
\small
\begin{tabular}{|l|l|}
\hline
\rowcolor[HTML]{EFEFEF}
\textbf{Symbol} & \textbf{Description} \\ \hline
$Q$ & Set of logical qubits (index $q$) \\ \hline
$P$ & Set of quantum computers (index $p$) \\ \hline
$T$ & Set of discrete time slots (index $t$) \\ \hline
$G$ & Set of gates in the circuit (index $g$) \\ \hline
$O_g$ & Set of operand qubits for gate $g$ \\ \hline
$s_p$ & Storage capacity at QC $p$ \\ \hline
$e_p$ & Execution capacity at QC $p$ \\ \hline
$c^s_p$ & Per-qubit leasing cost for storage at QC $p$ \\ \hline
$c^e_p$ & Per-qubit leasing cost for execution capacity at QC $p$ \\ \hline
$c^g_{p}$ &  Cost of execution of gate $g$ at QC $p$ \\ \hline
$E(p', p)$ & Cost of teleportation of a qubit from $p'$ to $p$ \\ \hline
$\beta$ & Makespan weight \\ \hline
$A_{p,g}$ & 1 if gate $g$ can execute at QC $p$ else 0 \\ \hline
$\tau_g$ & Duration of gate $g$  \\ \hline
$r^s_p$ & Number of storage qubits leased from $p$ \\ \hline
$r^e_p$ & Number of execution qubits leased from $p$ \\ \hline
$x_{q,p,t}$ & 1 if qubit $q$ is located at QC $p$ at time $t$ else 0\\ \hline
$z_{q,p,t}$ & 1 if qubit $q$ uses execution at $p$ at $t$ else 0\\ \hline
$w^t_{q,p',p}$ & 1 if qubit $q$ is teleported from $p'$ to $p$ at $t$ else 0\\ \hline
$u_{g,p}$ & 1 if gate $g$ executes at QC $p$ else 0\\ \hline
$t_g$ & Start time of gate $g$ \\ \hline
$T_{\max}$ & Makespan (circuit completion time) \\ \hline
\end{tabular}
\end{table}

We now describe the decision variables that appear in our optimization problem (see Table \ref{table:notation_params}). These variables fall into four categories: resource allocation, qubit placement, qubit relocation, and gate scheduling.

\textbf{Resource Allocation Variables}: Let the integer variable $r^s_p$ (respectively, $r^e_p$) denote the number of storage (respectively,  execution) qubits that the agent leases from QC $p \in P$. 

\textbf{Qubit Placement Variables}: These track qubit locations and resource usage over time. The binary variable $x_{q,p,t}$ indicates whether qubit $q$ is located at QC $p$ in time slot $t$, while $z_{q,p,t}$ indicates whether qubit $q$ actively uses execution capacity at QC $p$ in time slot $t$. The distinction between these variables allows us to model scenarios where a qubit may be stored at a QC without actively consuming execution resources.

\textbf{Qubit Relocation Variables}: These model qubit relocation via teleportation. The binary variable $w^t_{q,p',p}$ equals $1$ if qubit $q$ is teleported from $p'$ to $p$ in time slot $t$ using pre-shared entanglement. The cost of teleporting qubit \(q\) from \(p'\) to \(p\) at time \(t\) is modeled as
\[
\text{TeleportCost}(q,p',p,t) = E(p',p) \cdot w_{q,p',p}^{t}.
\]

\textbf{Gate Scheduling Variables}: They determine when and where gates execute. The binary variable $u_{g,p}$ indicates whether gate $g$ is assigned to QC $p$ for execution. The integer variable $t_g$ specifies the start time of gate $g$, and $T_{\max}$ represents the overall makespan (completion time) of the circuit execution.

\subsection{Problem Formulation}
\label{SSC:problem:formulation}
In this section, we formulate the JQLQCD problem, which is the focus of this paper. The objective of the agent is as follows.
\begin{align}
\min& 
  \underbrace{\sum_{p\in P} 
    \big( c^{s}_p\, r^{s}_{p} + c^{e}_p\, r^{e}_{p}\big)}_{\text{Leasing cost}}  + \underbrace{\sum_{g\in G}\sum_{p\in P} 
   c^{g}_{p}\, u_{g,p}}_{\text{Gate execution cost}} \notag\\
& + \underbrace{ 
   \sum_{q \in Q} \sum_{p' \in P}  \sum_{p \in P}\sum_{t \in T} E(p',p)\,w_{q,p',p}^{t}}_{\text{Teleportation cost}} + \beta\,T_{\max}. \label{eq:obj}
\end{align}
The first two terms correspond to the leasing and gate execution costs. The third term corresponds to the teleportation cost. The final term corresponds to the makespan.

The constraints of the problem, which are explained below, are as follows. 
\begin{align}
r^{s}_{p} &\in \{0,1,\ldots,s_p\}, \forall p, \label{EQ:resource:cap:constraint:1} \\
 r^{e}_{p} &\in \{0,1,\ldots, e_p\}, \forall p, \label{EQ:resource:cap:constraint:2} \\
\gamma_{q,p',p}^{t} &\in \{0,1\}, \forall q,p',p,t, \text{ where } p' \neq p, \label{EQ:binary:decision:variable:constraint:1} \\
w_{q,p',p}^{t} &\in \{0,1\}, \forall q,p',p,t, \text{ where } p' \neq p, \label{EQ:binary:decision:variable:constraint:2} \\
x_{q,p,t} &\in \{0,1\}, \forall q,p,t, \label{EQ:binary:decision:variable:constraint:3}\\
z_{q,p,t} &\in \{0,1\}, \forall q,p,t, \label{EQ:binary:decision:variable:constraint:4}\\
u_{g,p} &\in \{0,1\}, \forall g,p, \label{EQ:binary:decision:variable:constraint:5} \\
\sum_{q \in Q} x_{q,p,t} &\le r^{s}_{p}, \forall p,t, \label{EQ:resource:usage:constraint1} \\
\sum_{q \in Q} z_{q,p,t} &\le r^{e}_{p}, \forall p,t. \label{EQ:resource:usage:constraint2} \\
\sum_{p\in P} x_{q,p,t}&=1, \forall q,t, \label{EQ:qubit:location:uniqueness:constraint} \\
z_{q,p,t} &\le x_{q,p,t},  \forall q,p,t, \label{EQ:execution:requires:presence} \\
\sum_{p\in P} u_{g,p} &= 1, \forall g, \label{EQ:gate:assignment:constraint:1} \\
u_{g,p} &\le A_{p,g}, \forall g,p, \label{EQ:gate:assignment:constraint:2} \\
u_{g,p} &\le x_{q,p,t_g},  \forall q\in O_g,\, p,g. \label{EQ:operands:present} \\
t_{g_2} &\ge t_{g_1} + \tau_{g_1}, \forall g_1, g_2 \text{ such that } g_1 \prec g_2, \label{EQ:precedence:makespan:constraint:1} \\
T_{\max} &\ge t_g + \tau_{g}, \forall g, \label{EQ:precedence:makespan:constraint:2} \\
x_{q,p,t} - x_{q,p,t-1} &\le \max_{p' \in P,\, p' \neq p}  w_{q,p',p}^{t}, \forall q,p',p,t. \label{EQ:movement:consistency:constraint:2}
\end{align}

\textbf{Resource Capacity Constraints (\eqref{EQ:resource:cap:constraint:1} and \eqref{EQ:resource:cap:constraint:2}):}
Constraints \eqref{EQ:resource:cap:constraint:1} and \eqref{EQ:resource:cap:constraint:2} ensure that the numbers of storage and execution qubits leased from each QC $p$ must be non-negative and bounded by that QC's maximum storage and execution capacities $s_p$ and $e_p$, respectively.

\textbf{Binary Decision Variables  (\eqref{EQ:binary:decision:variable:constraint:1}, \eqref{EQ:binary:decision:variable:constraint:2}, \eqref{EQ:binary:decision:variable:constraint:3}, \eqref{EQ:binary:decision:variable:constraint:4}, and \eqref{EQ:binary:decision:variable:constraint:5}):}
Constraints \eqref{EQ:binary:decision:variable:constraint:1}, \eqref{EQ:binary:decision:variable:constraint:2}, \eqref{EQ:binary:decision:variable:constraint:3}, \eqref{EQ:binary:decision:variable:constraint:4}, and \eqref{EQ:binary:decision:variable:constraint:5} enforce the binary nature of the decision variables $\gamma^t_{q,p',p}$,  $w^t_{q,p',p}$, $x_{q,p,t}$, $z_{q,p,t}$, and $u_{g,p}$, respectively.

\textbf{Resource Usage Constraints (\eqref{EQ:resource:usage:constraint1} and \eqref{EQ:resource:usage:constraint2}):}
Constraints \eqref{EQ:resource:usage:constraint1} and \eqref{EQ:resource:usage:constraint2} enforce that in any time slot, the number of qubits stored at a QC cannot exceed the leased storage capacity, and the number of qubits being processed cannot exceed the leased execution capacity, respectively.

\textbf{Qubit Location Uniqueness (\eqref{EQ:qubit:location:uniqueness:constraint}):}
Each qubit $q$ must be located at exactly one QC at any given time $t$. 

\textbf{Execution Requires Presence (\eqref{EQ:execution:requires:presence}):}
If a qubit $q$ uses execution capacity at a QC $p$, it must be located there. 

\textbf{Gate Assignment Constraints (\eqref{EQ:gate:assignment:constraint:1} and \eqref{EQ:gate:assignment:constraint:2}):}
Constraint \eqref{EQ:gate:assignment:constraint:1} ensures that each gate $g$ is assigned to exactly one QC $p$. Constraint \eqref{EQ:gate:assignment:constraint:2} ensures that gates are only assigned to QCs 
that are capable of executing them.

\textbf{Operands Present (\eqref{EQ:operands:present}):}
When a gate $g$ executes at a QC $p$, all qubits $q$ in its operand set $O_g$ must be present at that QC at the gate's start time $t_g$. 

\textbf{Precedence and Makespan Constraints (\eqref{EQ:precedence:makespan:constraint:1} and \eqref{EQ:precedence:makespan:constraint:2}):}
Constraint \eqref{EQ:precedence:makespan:constraint:1} enforces circuit precedence relationships: If gate $g_1$ precedes $g_2$ (denoted by $g_1 \prec g_2$), then $g_2$ cannot start until $g_1$ completes. Constraint \eqref{EQ:precedence:makespan:constraint:2} ensures that the makespan $T_{\max}$ is the completion time of the last gate.

\textbf{Relocation Consistency Constraints (\eqref{EQ:movement:consistency:constraint:2}):}
Constraint \eqref{EQ:movement:consistency:constraint:2} enforces that if a qubit's location changes between time slots, a  teleportation operation must have occurred.

Collectively, the above constraints ensure a physically feasible circuit execution process in which resource capacity constraints are satisfied, gates execute in the correct order, and all operations follow the principles of quantum mechanics.

\begin{remark}
We make the following assumptions about the qubit relocation times. The general formulation assumes that qubit teleportation takes some non-negative amount of time, denoted by $\tau_{\text{teleport}}(p', p)$,  to complete. Specifically, we assume that these teleportation times are at most the width of a single time slot, i.e.,
\begin{equation}
\tau_{\text{teleport}}(p', p) \in [0, 1], \quad \forall\, p', p \in P,\ p' \neq p,
\end{equation}
so that any qubit teleportation completes within the same time slot in which it begins. This ensures consistency with the discrete time-slot formulation, in which the relocation consistency constraint \eqref{EQ:movement:consistency:constraint:2} implicitly captures the effect of teleportation on qubit location across consecutive slots. Some of the special cases in Section \ref{sec:special-cases} (specifically, those in Section \ref{subsec:homogeneous-zero-cost}) assume instantaneous teleportation ($\tau_{\text{teleport}}(p', p) = 0, \, \forall p', p$) to enable closed-form analysis. This assumption is explicitly stated when applicable and is discussed in detail in Section \ref{subsec:homogeneous-zero-cost}. For the general problem, $\tau_{\text{teleport}}(p', p)$ are problem-dependent parameters determined by the entanglement distribution rates and classical communication latencies for teleportation.
\end{remark}

\section{Computational Complexity}
\label{sec:complexity}

\subsection{NP-Completeness}
\label{SSC:NP:completenss:of:JQLQCD}
\begin{theorem}\label{thm:np-complete}
The JQLQCD problem is NP-complete.
\end{theorem}
\begin{IEEEproof}
A candidate solution to the JQLQCD problem specifies $r^s_p, r^e_p$,  $x_{q,p,t}$, $z_{q,p,t}$,  $u_{g,p}$, $t_g$, $\gamma^t_{q,p',p}$, and $w^t_{q,p',p}$ for all $q \in Q, p, p' \in P, t \in T$, and $g \in G$. Given a candidate solution, checking whether it satisfies the capacity constraints $\sum_q x_{q,p,t} \leq r^s_p$ and $\sum_q z_{q,p,t} \leq r^e_p$ for all $p, t$ requires $O(|Q| \cdot |P| \cdot T)$ time.
Checking the qubit uniqueness constraints $\sum_p x_{q,p,t} = 1$ for all $q, t$ can be done in $O(|Q| \cdot |P| \cdot T)$ time.
The gate assignment constraints $\sum_p u_{g,p} = 1$ for all $g$  can be checked in  $O(|G| \cdot |P|)$ time. The precedence constraints $t_{g_2} \geq t_{g_1} + \tau_{g_1}$ for all $g_1, g_2$  such that  $g_1 \prec g_2$ can be checked in $O(|E|)$ time, where $|E|$ is the number of precedence edges. The makespan constraint $T_{\max} \geq t_g + \tau_g$ for all $g$ can be checked in $O(|G|)$ time.  
The relocation consistency constraints \eqref{EQ:movement:consistency:constraint:2})  can be checked in $O(|Q| \cdot |P|^2 \cdot T)$ time. Finally, the 
total cost given in \eqref{eq:obj} can be computed in $O(|P| \cdot |G| + |Q| \cdot |P|^2 \cdot T)$ time. The total verification time is $O(|Q| \cdot |P|^2 \cdot T + |G| \cdot |P| + |E|)$, which is polynomial in the input size. Hence, the problem is in class NP \cite{kleinberg2006algorithm}. 

Now, we show that  the multiprocessor scheduling problem with precedence constraints, denoted $P|\text{prec}|C_{\max}$ in standard scheduling notation, which is known to be NP-complete \cite{Ullman1975, GareyJohnson1979, Lenstra1977}, is polynomial-time reducible to the (decision version of the) JQLQCD problem. Consider the following instance of the $P|\text{prec}|C_{\max}$ problem:  We are given a set of tasks $\mathcal{T} = \{T_1, T_2, \ldots, T_n\}$ with processing time $\pi_i$ for task $T_i$, a set of $m$ identical parallel processors, a set of precedence constraints $T_i \prec T_j$ forming a DAG, and a target makespan bound $K$, where the makespan is defined as the completion time of the last task. The problem is to decide whether there exists a schedule with makespan at most $K$. This problem was shown to be strongly NP-hard in \cite{Ullman1975} and it remains NP-hard even when all tasks have unit processing time ($\pi_i = 1$ for all $i$) and the precedence graph is restricted to certain structures (e.g., trees with bounded degree, series-parallel graphs).

From the above instance of the $P|\text{prec}|C_{\max}$ problem, we construct an instance of the JQLQCD problem as follows.  The quantum circuit is created by generating a gate $g_i$ with duration $\tau_{g_i} = \pi_i$ for each task $T_i$. We introduce $n$ qubits $Q = \{q_1, q_2, \ldots, q_n\}$, where each gate $g_i$ operates exclusively on its corresponding qubit $q_i$ as a single-qubit gate. The precedence constraints are mapped directly: If $T_i \prec T_j$ in the scheduling instance, then we enforce $g_i \prec g_j$ in the circuit. We configure the quantum circuit created as explained above with the following parameters. We create $|P| = m$ QCs-- one for each processor. Each QC has infinite capacity, so we set $s_p = e_p = \infty$ for all $p \in P$. All costs are set to zero:
\begin{align}
c^s_p &= c^e_p = 0, \quad \forall p \in P, \\
c^g_{p} &= 0, \quad \forall p \in P, g \in G, \\
E(p', p) &= 0, \quad \forall p, p' \in P.
\end{align}
The availability of gates is unrestricted, with $A_{p,g} = 1$ for all $p \in P, g \in G$, i.e., all gates can run on all QCs. Finally, we set the makespan weight $\beta = 1$, ensuring that the objective function in \eqref{eq:obj} equals the makespan $T_{\max}$.  

Note that this reduction can be performed in polynomial-time as it involves creating $n$ gates and $n$ qubits, which requires $O(n)$ time, copying precedence constraints, which requires $O(|E|)$ time, creating $m$ QCs, which requires $O(m)$ time, and setting cost matrices to zero, which requires $O(m^2+m \cdot n)$ time, yielding an overall complexity of $O(m^2 + mn + |E|)$. 

In the above instance of the JQLQCD problem, we ask: Does there exist a solution with makespan $T_{\max} \leq K$? We will now prove that the instance of $P|\text{prec}|C_{\max}$ has a schedule with makespan at most $K$ if and only if the instance of the JQLQCD problem has a solution with makespan $T_{\max} \leq K$.

To prove necessity, suppose the instance of $P|\text{prec}|C_{\max}$ has a schedule $\mathcal{S}$ with makespan at most $K$. We map this schedule to  a solution to the JQLQCD problem as follows. For each task $T_i$ scheduled on processor $j$ with a start time of $t$ in $\mathcal{S}$, we assign gate $g_i$ to QC $p_j$ by setting $u_{g_i, p_j} = 1$, set the gate start time $t_{g_i} = t$, and place qubit $q_i$ at QC $p_j$ for all relevant time slots by setting $x_{q_i, p_j, t'} = 1$ for $t' \in [t, t + \pi_i]$. No qubit teleportations are needed, so $w^t_{q,p',p} = 0$ for all $q, p, p', t$. The resulting makespan is $T_{\max} = \max_{g \in G} (t_g + \tau_{g})$. Since $\mathcal{S}$ satisfies all precedence constraints and has makespan at most $K$, the constructed solution to the JQLQCD problem also satisfies all precedence constraints (see \eqref{EQ:precedence:makespan:constraint:1}) and has objective function value $T_{\max} \leq K$ (see \eqref{EQ:precedence:makespan:constraint:2} and \eqref{eq:obj}). All capacity constraints are trivially satisfied because $s_p = e_p = \infty$. Thus, the JQLQCD problem has a feasible solution with makespan $T_{\max} \leq K$, which shows necessity.

To prove sufficiency, suppose the JQLQCD problem instance has a feasible solution with makespan $T_{\max} \leq K$. We construct a schedule $\mathcal{S}$ for the $P|\text{prec}|C_{\max}$ problem as follows. For each gate $g_i$ assigned to QC $p_j$ (i.e., $u_{g_i, p_j} = 1$) with start time $t_{g_i}$, we schedule task $T_i$ on processor $j$ with a start time of $t_{g_i}$. This schedule $\mathcal{S}$ is feasible because the solution to the  JQLQCD problem instance satisfies the gate assignment constraint (see \eqref{EQ:gate:assignment:constraint:1}), ensuring that each gate is assigned to exactly one QC, and thus each task is assigned to exactly one processor; also, the precedence constraints (see \eqref{EQ:precedence:makespan:constraint:1})  guarantee that if $g_i \prec g_j$, then $t_{g_j} \geq t_{g_i} + \tau_{g_i}$, which translates to $T_i \prec T_j$ being satisfied in $\mathcal{S}$. The makespan constraint (see \eqref{EQ:precedence:makespan:constraint:2}) ensures that $T_{\max} \geq t_g + \tau_g$ for all $g$, so the makespan of $\mathcal{S}$ is $T_{\max} \leq K$. Therefore, $\mathcal{S}$ is a valid schedule with makespan at most $K$, which shows sufficiency. The result follows. 
\end{IEEEproof}

In fact, in the above proof of Theorem \ref{thm:np-complete}, we have shown the following stronger result. 
\begin{corollary}
\label{CR:NP:completeness:special:case}
The JQLQCD problem remains NP-complete even when
\begin{enumerate}[label=(\alph*)]
\item all QCs have infinite capacity ($s_p = e_p = \infty$),
\item all costs except makespan are zero ($c^s_p = c^e_p = c^g_{p} = E(p',p) = 0$ for all $q, p, p'$),
\item all gates can execute on all QCs ($A_{p,g} = 1$ for all $p, g$), and
\item the circuit consists only of single-qubit gates.
\end{enumerate}
\end{corollary}
\begin{IEEEproof}
This result follows directly from the reduction used to prove Theorem \ref{thm:np-complete}, which uses exactly these restricted conditions.
\end{IEEEproof}

Recall that a problem is said to be strongly NP-complete if it remains NP-complete even when all numerical parameters are bounded by a polynomial in the input size \cite{Ullman1975, Lenstra1977}.

\begin{corollary} 
\label{CR:strong:NP:hard}
The JQLQCD problem is strongly NP-complete. 
\end{corollary}
\begin{IEEEproof}
The $P|\text{prec}|C_{\max}$ problem is known to be strongly NP-complete \cite{Ullman1975, Lenstra1977}. We show that the reduction used to prove Theorem \ref{thm:np-complete} can be carried out with all numerical parameters polynomially bounded in the input size. Recall that in the reduction, the gate durations are set to $\tau_{g_i} = \pi_i$, which are polynomially bounded whenever the processing times $\pi_i$ are. All costs are set to zero, which is trivially polynomially bounded. The storage and execution capacities, which were set to $\infty$ in the proof of Theorem~\ref{thm:np-complete}, can be replaced by $M = |Q| \cdot |G|$, since no feasible solution ever needs to store more than $|Q|$ qubits at any QC or process more than $|Q|$ qubits in any time slot; thus, setting $s_p = e_p = M$ for all $p \in P$ is functionally equivalent to unlimited capacity within any feasible solution, while keeping all parameters polynomially bounded in the input size. Hence, the reduction is a valid polynomial-time reduction from a strongly NP-complete problem in which all numerical parameters are polynomially bounded, and the result follows.
\end{IEEEproof}
Strong NP-completeness implies that no pseudo-polynomial-time algorithm for finding the optimal solution of the JQLQCD problem exists (unless $\textsf{P} = \textsf{NP}$), ruling out dynamic programming (DP) approaches that work for polynomially bounded numerical parameter values.

\subsection{Parameterized Complexity}
While the general JQLQCD problem is NP-complete, we now show that it is tractable for some specific parameter ranges. Natural parameters to consider include the circuit width $w_c = \max_{g \in G} |O_g|$ (maximum number of qubits in any gate's operand set), the circuit depth $d$ (length of the longest path of dependent gates from input to output, i.e., critical path), the number of QCs $|P|$, the tree-width $w_t$ of the circuit dependency graph \cite{Robertson1986, Bodlaender1998}, and the maximum QC capacity $\kappa = \max_p \min(s_p, e_p)$.

\begin{theorem}
The JQLQCD problem is fixed-parameter tractable (FPT) \cite{Downey2013, Cygan2015} when parameterized by the number of QCs $|P|$ and the maximum QC capacity $\kappa$, assuming that the circuit has bounded tree-width $w_t$, $|P|$, and $\kappa$ \cite{Robertson1986, Bodlaender1998}.
\end{theorem}
\begin{IEEEproof}
With bounded tree-width $w_t$ and bounded $|P|$ and $\kappa$, we construct a tree decomposition of the circuit's dependency graph and apply DP over it. Each node in the tree decomposition contains a \emph{bag}-- a subset of at most $w_t + 1$ qubits that must be considered together.

\textbf{State encoding.} For each bag, the DP state records: (i) the assignment of each qubit in the bag to a QC, i.e., a mapping from qubits to $P$, which encodes $x_{q,p,t}$; (ii) the current execution capacity usage at each QC, which encodes $z_{q,p,t}$ and allows us to determine $r^e_p = \max_t \sum_q z_{q,p,t}$; (iii) the gate assignment $u_{g,p}$ for each gate $g$ whose operand qubits are all present in the current bag; and (iv) the earliest feasible start time $t_g$ for each such gate, respecting the precedence constraints \eqref{EQ:precedence:makespan:constraint:1} within the bag. The leasing variables $r^s_p$ and $r^e_p$ are not independent decisions but are derived from the qubit placement and execution assignments as $r^s_p = \max_t \sum_q x_{q,p,t}$ and $r^e_p = \max_t \sum_q z_{q,p,t}$, so they are determined once the placement is fixed.

\textbf{Transitions.} When two child bags are merged at a parent bag, the DP combines their partial solutions by: checking that shared qubits have consistent QC assignments across the two children; computing qubit teleporation costs for qubits that change QC between bags; and propagating the earliest feasible start times for gates in the parent bag respecting inter-bag precedence constraints \eqref{EQ:precedence:makespan:constraint:1}.

\textbf{Optimality.} The DP computes, for each bag and each state, the minimum partial cost (leasing cost, gate execution cost, and teleportation cost) over all feasible assignments of qubits and gates within that subtree of the decomposition. Since the tree decomposition ensures that all interactions between qubits separated across bags are mediated through shared qubits appearing in the bag boundary, the optimal global solution is obtained by combining optimal partial solutions bottom-up, with the root bag yielding the globally optimal assignment and the makespan $T_{\max}$ computed as $\max_{g \in G}(t_g + \tau_g)$ per constraint \eqref{EQ:precedence:makespan:constraint:2}.

\textbf{Complexity.} The state space per bag is $O(|P|^{w_t} \cdot \kappa^{|P| \cdot w_t})$, tracking the QC assignment of each qubit in the bag and the capacity usage at each QC. With $O(|G|)$ bags in the tree decomposition and $O(|P|^{2w_t})$ transitions per bag, the total complexity is $O(|G| \cdot |P|^{3w_t} \cdot \kappa^{|P| \cdot w_t})$, which is polynomial in $|G|$ for fixed $|P|$, $w_t$, and $\kappa$, establishing FPT tractability. The result follows.
\end{IEEEproof}

For typical quantum circuits, $w_t \leq 2$ for nearest-neighbor architectures with limited qubit connectivity, $w_c \leq 2$ for circuits with only single and two-qubit gates, and $|P|$ is typically small (e.g., 2-10 distributed QCs). These observations suggest that practical instances may be more tractable than the worst-case complexity suggests, motivating the development of efficient algorithms for special cases with restricted parameter ranges. In Section \ref{sec:special-cases}, we propose polynomial-time algorithms for several special cases.

\section{Polynomial-Time Optimal Algorithms for Special Cases}
\label{sec:special-cases}
Although Theorem \ref{thm:np-complete} shows that the general JQLQCD problem is NP-hard, in this section, we identify several special cases in which the problem can be solved optimally in closed-form or via polynomial-time algorithms.

\subsection{Case 1: Unlimited Capacity QC in a Heterogeneous Network}
\label{subsec:single-qc}

\textbf{Problem Setting:} The quantum network consists of one QC \(p_0 \in P\) with unlimited capacity (\(s_{p_0} = e_{p_0} = \infty\)) and a set of additional QCs \(P' = P \setminus \{p_0\}\) with finite capacities \(s_p, e_p < \infty\) for all \(p \in P'\). The QCs are connected in an arbitrary network topology with heterogeneous teleportation costs \(E(p',p)\), leasing costs \(c^{s}_p, c^{e}_p\), and gate execution costs \(c^{g}_{p}\). We refer to the QCs in $P'$ as ``limited QCs'' since their capacities are finite.

\textbf{Base Solution:} When all qubits of the given quantum circuit are statically assigned to \(p_0\), no teleportation is required. The problem reduces to gate scheduling with precedence constraints, whose optimal solution can be found by topological sorting of the circuit DAG to determine valid gate orderings \cite{Cormen2009, Knuth1997} and applying a list scheduling algorithm \cite{Graham1969, Coffman1976} to minimize the makespan $T_{\max}$.
The complexity of this algorithm is \(O(|G| + |E|)\), where \(E\) denotes the set of precedence edges in the circuit. The achieved cost (value of the objective function in \eqref{eq:obj}) is  
\begin{equation}
C^{*}_{\text{centralized}} = c^{s}_{p_0} \cdot |Q| + c^{e}_{p_0} \cdot |Q| + \sum_{g \in G} c^{g}_{p_0} + \beta \cdot T_{\max}^{\text{opt}},
\end{equation}
where \(T_{\max}^{\text{opt}}\) is the critical path length of the circuit DAG \cite{Cormen2009}. 

Despite the availability of other QCs, viz., those in \(P'\), the above centralized solution (statically allocating all qubits to \(p_0\)) may still be optimal. In Section \ref{SSSC:sufficient:conditions:for:optimality:of:centralized:solution}, we present multiple sets of sufficient conditions under which this occurs.

\begin{remark}
\label{rem:qubit-location}
We make the following assumptions about qubit locations. Before the start of 
circuit execution, all qubits are located at $p_0$. After the circuit execution 
completes, all qubits must be returned to $p_0$. These assumptions are natural 
in the cloud computing setting where $p_0$ is the agent's home QC: qubits 
originate at $p_0$, may be temporarily teleported to QCs in $P'$ for cheaper 
execution, and must be returned to $p_0$ upon completion. Under these 
assumptions, any distributed solution that teleports qubits to QCs in $P'$ incurs 
round-trip teleportation costs. Note that a distributed solution may still keep some 
qubits entirely at $p_0$ throughout execution (incurring no teleportation cost for 
those qubits), while teleporting other qubits to QCs in $P'$ for some gates and 
returning them to $p_0$ afterward.
\end{remark}

\subsubsection{Sufficient Conditions for Optimality of Centralized Solution}
\label{SSSC:sufficient:conditions:for:optimality:of:centralized:solution}
The centralized solution is optimal if \emph{any} of the conditions S1 to S4 given below holds.

\noindent\textbf{Condition S1 (Dominant Teleportation Cost):}
The minimum teleportation cost to any QC in $P'$ exceeds the total potential savings 
from both lower gate execution costs and lower leasing costs at QCs in $P'$ 
compared to $p_0$:
\begin{equation}
\label{eq:cond-s1}
\min_{p \in P'} E(p_0, p) > \Delta C_{\text{gate}} + \Delta C_{\text{lease}},
\end{equation}
where
\begin{align}
\Delta C_{\text{gate}} &= \frac{1}{|Q|} \sum_{g \in G} \max_{p \in P'} \big(c^{g}_{p_0} - c^{g}_{p}\big)^{+}, \notag \\
\Delta C_{\text{lease}} &= \max_{p \in P'} \big[(c^{s}_{p_0} + c^{e}_{p_0}) - (c^{s}_p + c^{e}_p)\big]^{+}, \notag
\end{align}
and $x^{+} = \max(0, x)$.

Intuitively, $\Delta C_{\text{gate}}$ captures the average per-qubit savings in gate 
execution costs from running gates at cheaper QCs in $P'$ instead of $p_0$, and 
$\Delta C_{\text{lease}}$ captures the maximum per-qubit savings in leasing costs 
from storing and processing qubits at a cheaper QC in $P'$ instead of $p_0$. 
Condition \eqref{eq:cond-s1} ensures that even when both sources of savings are 
combined, the teleportation cost required to exploit them exceeds the total benefit, 
making the centralized solution at $p_0$ optimal.

\noindent\textbf{Condition S2 (Prohibitive Teleportation Cost):}
All teleportation costs from \(p_0\) to any other QC are sufficiently high:
\begin{equation}
\label{eq:cond-s2}
\min_{p \in P'}  E(p_0, p)  > \Delta C_{\text{gate}} + \Delta C_{\text{lease}},
\end{equation}
where
\begin{align}
\Delta C_{\text{gate}} &= \max_{g \in G} \big(c^{g}_{p_0} - \min_{p \in P'} c^{g}_{p}\big)^{+}, \notag \\
\Delta C_{\text{lease}} &= \max_{p \in P'} \big[(c^{s}_{p_0} + c^{e}_{p_0}) - (c^{s}_p + c^{e}_p)\big]^{+}, \notag
\end{align}
\textbf{AND} the centralized solution is no worse than executing the circuit using only the QCs in \(P'\):
\begin{equation}
\label{eq:cond-s2-isolated}
C^{*}_{\text{centralized}} \leq C^{*}_{\text{isolated}},
\end{equation}
where \(C^{*}_{\text{isolated}}\) is a lower bound on the cost of executing the circuit using only QCs in \(P'\) (i.e., without \(p_0\)):
\begin{equation}
\label{eq:isolated-lower-bound}
C^{*}_{\text{isolated}} \leq \sum_{p \in P'} (c^s_p + c^e_p) \cdot \min\left(|Q|, \min(s_p, e_p)\right) + \sum_{g \in G} \min_{p \in P'} c^g_{p}.
\end{equation}
Intuitively, the first part, \eqref{eq:cond-s2}, ensures that hybrid solutions (using both \(p_0\) and QC(s) in \(P'\)) are suboptimal by enforcing that teleporting qubits between \(p_0\) and any limited QC costs more than the potential savings in gate execution costs and leasing costs combined. The second part, \eqref{eq:cond-s2-isolated}, ensures that centralized execution at \(p_0\) is also better than circuit execution without using \(p_0\) at all. Together, the two parts guarantee global optimality of the centralized solution.

\noindent\textbf{Condition S3 (Capacity Bottleneck):}
The total network capacity excluding \(p_0\) is insufficient to store all qubits:
\begin{equation}
\label{eq:cond-s4a}
\sum_{p \in P'} \min(s_p, e_p) < |Q|,
\end{equation}
and the maximum total savings achievable by teleporting any subset of qubits to QCs 
in $P'$-- from both lower leasing costs and lower gate execution costs-- do 
not exceed the minimum teleportation cost incurred by doing so:
\begin{align}
\label{eq:cond-s4c}
&\max_{p \in P'}\big[(c^{s}_{p_0} + c^{e}_{p_0}) - (c^{s}_p + c^{e}_p)\big]^{+} 
\cdot \sum_{p \in P'} \min(s_p, e_p) \notag \\
&+ \sum_{g \in G} \max_{p \in P'} \big(c^{g}_{p_0} - c^{g}_{p}\big)^{+} \notag \\
&\leq \sum_{p \in P'} \min(s_p, e_p) \cdot 
\min_{p \in P'} E(p_0, p),
\end{align}
where the first term on the LHS is the maximum leasing cost 
savings from teleporting up to $\sum_{p \in P'} \min(s_p, e_p)$ qubits to QCs in 
$P'$, the second term on the LHS is the maximum gate execution cost savings from 
running gates at cheaper QCs in $P'$, and the RHS is the minimum total teleportation 
cost of sending those qubits to $P'$.

Intuitively, since $P'$ lacks sufficient aggregate capacity to store all $|Q|$ 
qubits, at least $k = |Q| - \sum_{p \in P'} \min(s_p, e_p)$ qubits must reside 
at $p_0$. However, the remaining qubits that \emph{could} fit in $P'$ need not 
be teleported there: they may simply be processed entirely at $p_0$ without any 
teleportation. Conditions \eqref{eq:cond-s4a} and \eqref{eq:cond-s4c} together ensure 
that even if we were to teleport as many qubits as possible to $P'$ to exploit lower 
leasing and gate execution costs, the combined savings from doing so would not 
offset the teleportation costs incurred. Hence, all qubits are best kept at $p_0$, 
making centralized execution optimal. Here, \(\min(s_p, e_p)\) represents the 
effective usable capacity of QC \(p\).
\noindent\textbf{Condition S4 (Communication Bottleneck):}
The following two conditions together guarantee that centralized execution at 
$p_0$ is globally optimal. First, the cost of centralized execution at $p_0$ 
is no worse than any solution that assigns some qubits entirely to QCs in $P'$ 
with no teleportation between $p_0$ and $P'$:
\begin{align}
\label{eq:cond-s5}
&c^{s}_{p_0} + c^{e}_{p_0} + \frac{1}{|Q|} \sum_{g \in G} c^{g}_{p_0} \notag \\
&\leq \min_{p \in P'} \left(c^{s}_p + c^{e}_p\right) 
+ \frac{1}{|Q|} \sum_{g \in G} \min_{p \in P'} c^{g}_{p},
\end{align}
where the LHS is the per-qubit cost of centralized execution at $p_0$ 
(leasing plus average gate execution cost), and the RHS is the minimum 
per-qubit leasing cost over all QCs in $P'$ plus the average per-gate 
minimum execution cost over all QCs in $P'$ (noting that different gates 
may be executed at different QCs in $P'$ in the optimal solution).

Second, the minimum teleportation cost between $p_0$ and any QC in $P'$ exceeds 
the maximum per-qubit savings achievable by teleporting a qubit from $p_0$ to 
some QC in $P'$:
\begin{align}
\label{eq:cond-s5-movement}
&\min_{p \in P'} E(p_0,p) > \notag \\
& \max_{p \in P'} 
\Big[\big(c^{s}_{p_0} + c^{e}_{p_0}\big) - \big(c^{s}_p + c^{e}_p\big)\Big]^{+} 
+ \frac{1}{|Q|} \sum_{g \in G} \max_{p \in P'} 
\big(c^{g}_{p_0} - c^{g}_{p}\big)^{+},
\end{align}
where the first term on the RHS is the maximum per-qubit 
leasing cost saving from teleporting a qubit to the cheapest QC in $P'$, and the 
second term is the average per-gate execution cost saving from running each 
gate at its cheapest QC in $P'$.

Intuitively, there are two types of potentially cheaper solutions to rule out: 
(i) \emph{Hybrid solutions with no teleportation}, in which some qubits are 
processed entirely at $p_0$ and others entirely at QCs in $P'$ with no 
inter-QC teleportation. Condition \eqref{eq:cond-s5} rules these out by ensuring 
that the per-qubit-and-gate cost at $p_0$ is no worse than the best achievable 
by splitting execution across $P'$, accounting for the fact that different gates 
may be run at different QCs in $P'$. (ii) \emph{Hybrid solutions with teleporation}, 
in which qubits are teleported between $p_0$ and QCs in $P'$ to exploit lower costs 
there. Condition \eqref{eq:cond-s5-movement} rules these out by ensuring that 
the teleportation cost exceeds the combined leasing and gate execution cost savings 
from such teleportation. Together, \eqref{eq:cond-s5} and 
\eqref{eq:cond-s5-movement} guarantee that centralized execution at $p_0$ is 
globally optimal.

\subsubsection{Necessary Conditions for Optimality of Distributed Solution}

We now present multiple conditions, N1 to N3, each of which is necessary for a 
distributed solution to outperform pure centralization at $p_0$. These conditions 
are stated under the qubit location assumptions of Remark~\ref{rem:qubit-location}.

\noindent\textbf{Condition N1 (Cost Advantage From Partial Distribution):}
There exists at least one QC $p \in P'$ and gate $g \in G$ such that the total 
cost of executing $g$ at $p$-- including the leasing cost at $p$, the 
round-trip teleportation cost from $p_0$ to $p$ and back, and the gate execution 
cost at $p$-- is strictly less than the total cost of executing $g$ at $p_0$ 
including its leasing cost:
\begin{equation}
\label{eq:cond-n1}
c^{g}_{p} + (c^{s}_p + c^{e}_p) + 2 \cdot E(p_0,p) < 
c^{g}_{p_0} + (c^{s}_{p_0} + c^{e}_{p_0}).
\end{equation}
Intuitively, for any distributed solution to be beneficial, there must exist at 
least one gate $g$ and QC $p \in P'$ such that executing $g$ at $p$ (including 
the leasing cost at $p$ and paying for the round-trip teleportation of the operand qubits from $p_0$ 
to $p$ and back) is cheaper than executing $g$ at $p_0$ (including its leasing 
cost). The round-trip teleportation is necessary because, by 
Remark~\ref{rem:qubit-location}, all qubits start and end at $p_0$.

\noindent\textbf{Condition N2 (Sufficient Savings From Distribution):}
There exists a subset of gates $G' \subseteq G$ and an assignment of these gates 
to QCs in $P'$ such that the total savings from executing $G'$ at QCs in $P'$ 
instead of $p_0$ exceed the total costs incurred by doing so:
\begin{align}
\label{eq:cond-n2}
&\sum_{g \in G'} \big(c^{g}_{p_0} - \min_{p \in P'} c^{g}_{p}\big) 
+ \sum_{q \in Q(G')} (c^{s}_{p_0} + c^{e}_{p_0}) \notag \\
&> \sum_{q \in Q(G')} 2 \cdot \min_{p \in P'} E(p_0, p) \notag \\
& \quad + \sum_{p \in P'} (c^s_p + c^e_p) \cdot |Q_p|,
\end{align}
where $Q(G') = \bigcup_{g \in G'} O_g$ is the set of qubits involved in the 
gates in $G'$, and $Q_p \subseteq Q(G')$ denotes the qubits assigned to QC $p$ 
with $\bigcup_{p \in P'} Q_p = Q(G')$.

Intuitively, the LHS is the total savings from running gates $G'$ at cheaper 
QCs in $P'$ instead of $p_0$, plus the leasing cost savings at $p_0$ for the 
qubits in $Q(G')$ that are teleported away from $p_0$. The RHS is the total 
round-trip teleportation cost for the qubits in $Q(G')$ (since by 
Remark~\ref{rem:qubit-location} they must return to $p_0$) plus the leasing 
costs at the QCs in $P'$ where those qubits are processed. For a distributed 
solution to be beneficial, the savings must exceed these costs.

\noindent\textbf{Condition N3 (Network Capacity Requirement):}
The network must have sufficient aggregate capacity to store and execute at 
least one gate:
\begin{equation}
\label{eq:cond-n3}
\sum_{p \in P'} \min(s_p, e_p) \geq \min_{g \in G} |O_g|.
\end{equation}
Intuitively, the network must be able to store and process at least the operand 
qubits of the gate with the fewest operands in the circuit. Without this minimal 
capacity, no distributed execution is possible regardless of cost considerations.

\subsubsection{Necessary and Sufficient Condition}
For the restricted case where all gate execution costs are identical (\(c^{g}_{p} = c^g\) for all \(p,g\)), we now provide a necessary and sufficient condition for 
the centralized solution (static allocation of all qubits to \(p_0\)) to be optimal.

\begin{theorem}
\label{thm:ns-homogeneous}
Suppose \(c^{g}_{p} = c^g\) for all \(p,g\), and that the following condition holds:
\begin{equation}
\label{eq:p0-used}
c^{s}_{p_0} + c^{e}_{p_0} \leq \min_{p \in P'} \left(c^{s}_p + c^{e}_p\right),
\end{equation}
which ensures that \(p_0\) is the cheapest QC for qubit storage and execution, so that the optimal solution necessarily stores at least one qubit at \(p_0\). Under this condition, the centralized solution at \(p_0\) is optimal if and only if
\begin{equation}
\label{eq:ns-homogeneous}
\min_{p \in P'} \Big[(c^{s}_p + c^{e}_p) - (c^{s}_{p_0} + c^{e}_{p_0}) + 2 E(p_0,p)\Big] \geq 0.
\end{equation}
\end{theorem}

The condition in \eqref{eq:ns-homogeneous} states that for every QC \(p \in P'\), the additional per-qubit leasing cost at \(p\) compared to \(p_0\), plus the round-trip teleportation cost, must be non-negative. If this holds, teleporting any qubits to \(P'\) cannot reduce total cost, making centralized execution optimal. Conversely, if some QC in \(P'\) has sufficiently lower leasing costs to offset teleportation costs (making the expression negative), then partial distribution becomes beneficial.

\begin{IEEEproof}[Proof of Theorem \ref{thm:ns-homogeneous}]
With homogeneous gate costs (\(c^g_{p} = c^g\) for all \(p,g\)), gate execution costs are identical regardless of where gates execute, so the only potential cost savings from partial distribution can come from reduced per-qubit leasing costs. Condition \eqref{eq:p0-used} ensures that \(p_0\) has the lowest per-qubit leasing cost among all QCs. Therefore, in any cost-minimizing solution, it is never beneficial to store all qubits exclusively at QCs in \(P'\) while leaving \(p_0\) unused; the optimal solution necessarily stores at least one qubit at \(p_0\).

For any qubit \(q\) to benefit from teleporting to some QC \(p \in P'\), the leasing cost savings at \(p\) must exceed the round-trip teleportation cost:
\begin{equation*}
(c^{s}_{p_0} + c^{e}_{p_0}) - (c^{s}_p + c^{e}_p) > 2 E(p_0,p).
\end{equation*}
The factor of $2$ accounts for round-trip teleportation, i.e., for teleporting the qubit from \(p_0\) to \(p\), and eventually returning to \(p_0\) (since \(p_0\) stores at least one qubit, there must be gates executed at \(p_0\), requiring qubits to return).

Rearranging, centralized execution at \(p_0\) is optimal if and only if no QC in \(P'\) offers sufficient cost advantage:
\begin{equation*}
(c^{s}_p + c^{e}_p) - (c^{s}_{p_0} + c^{e}_{p_0}) + 2 E(p_0,p) \geq 0, \, \forall p \in P'.
\end{equation*}
Taking the minimum over all \(p \in P'\) gives condition \eqref{eq:ns-homogeneous}.
\end{IEEEproof}

\subsubsection{Computational Complexity of Verification}
Verification of the sufficient conditions S1-S4 and necessary conditions N1-N3 involves computing cost differences and capacities in $O(|P| \cdot |G|)$ time, computing $C^{*}_{\text{centralized}}$ via topological sort in $O(|G| + |E|)$ time, computing the lower bound $C^{*}_{\text{isolated}}$ in $O(|P| \cdot |G|)$ time, and gate availability matrix checks in $O(|P| \cdot |G|)$ time. The total complexity is $O(|P| \cdot |G| + |E|)$. However, since the JQLQCD problem is NP-complete, this verification is significantly more efficient than determining if a distributed solution is optimal without using the necessary or sufficient conditions and by solving the full ILP. Algorithm \ref{alg:centralized-check} provides a practical decision procedure that leverages these sufficient and necessary conditions to avoid expensive ILP solving when centralized execution is provably optimal.

\begin{algorithm}[H]
\caption{Verification  of Sufficient and Necessary Conditions}
\label{alg:centralized-check}
\begin{algorithmic}[1]
\State \textbf{Input:} Instance with $p_0$ (unlimited QC), $P'$ (limited QCs)
\State Compute $C^{*}_{\text{centralized}}$ via gate scheduling on $p_0$
\State Compute lower bound $C^{*}_{\text{isolated}}$ for execution on $P'$ only
\If{any of S1--S4 holds}
    \Return Optimal centralized solution
\ElsIf{any of N1--N3 fails}
    \Return Optimal centralized solution
\Else
    \Return Solve full ILP
\EndIf
\end{algorithmic}
\end{algorithm}

\subsubsection{Example}
We now provide an example in which the centralized solution is optimal. Suppose 
$p_0$ has $c^{s}_{p_0} = c^{e}_{p_0} = 0$ (free unlimited capacity), 
$P' = \{p_1, p_2\}$ with $c^{s}_{p_i} = c^{e}_{p_i} = 10$, teleportation 
costs $E(p_0, p_i) = 50$ for $i \in \{1,2\}$, and all gate 
costs are identical at $c^{g}_{p} = 1$. We now show that Condition S2 holds. 
First, we verify \eqref{eq:cond-s2}. We have $\Delta C_{\text{gate}} = 0$ 
since all gate costs are identical, and 
\begin{align*}
\Delta C_{\text{lease}} &= \max_{p \in P'} \big[(c^{s}_{p_0} + c^{e}_{p_0}) 
- (c^{s}_p + c^{e}_p)\big]^{+} \\
&= \max(0 - 20, 0)^{+} \\
&= 0,
\end{align*}
since the leasing costs at $P'$ are higher than at $p_0$. So, condition 
\eqref{eq:cond-s2} becomes
\[
50 > 0 + 0 = 0,
\]
which is true. This confirms that hybrid solutions involving teleportation between 
$p_0$ and $P'$ are not cost-effective. We now verify \eqref{eq:cond-s2-isolated}. 
Note that
\begin{align*}
C^{*}_{\text{centralized}} &= 0 \cdot |Q| + 0 \cdot |Q| + |G| \cdot 1 
+ \beta \cdot T^{\text{opt}}_{\max} \\
&= |G| + \beta \cdot T^{\text{opt}}_{\max}, \\
C^{*}_{\text{isolated}} &\geq (c^s_{p_1} + c^e_{p_1}) \cdot \min(|Q|, 
\min(s_{p_1}, e_{p_1})) \\
&\quad + (c^s_{p_2} + c^e_{p_2}) \cdot \min(|Q|, \min(s_{p_2}, e_{p_2})) 
+ |G| \cdot 1 \\
&= 20 \cdot \min(|Q|, \min(s_{p_1}, e_{p_1})) \\
&\quad + 20 \cdot \min(|Q|, \min(s_{p_2}, e_{p_2})) + |G|.
\end{align*}
For $|Q| \geq 1$,
\[
C^{*}_{\text{isolated}} \geq 20 + |G| > |G| + \beta \cdot T^{\text{opt}}_{\max} 
= C^{*}_{\text{centralized}},
\]
where the last inequality holds when $\beta \cdot T^{\text{opt}}_{\max} \leq 20$, 
which is satisfied, e.g., when $\beta = 0$ or when the circuit is shallow enough 
that $T^{\text{opt}}_{\max} \leq \frac{20}{\beta}$. Since both parts of 
Condition S2 are satisfied, the centralized solution is optimal.

\subsection{Case 2: Homogeneous QCs with Zero Teleportation Cost}
\label{subsec:homogeneous-zero-cost}

\textbf{Problem Setting:} All QCs are identical (\(s_p = s\), \(e_p = e\), \(c^{s}_p = c^s\), \(c^{e}_p = c^e\), \(c^{g}_p = c^g\), \(A_{p,g} = 1\) for all \(p,g\)), and \(E(p',p) = 0\) for all \(p' \neq p\).

\textbf{Solution Strategy:} Teleportation is free, so qubits can be relocated 
without cost. The optimal strategy depends on the weight \(\beta\) of the 
makespan term in \eqref{eq:obj}. In Sections \ref{SSSC:beta:eq:0} and 
\ref{SSSC:beta:gt:0}, we consider the cases in which \(\beta = 0\) and 
\(\beta > 0\), respectively.

\subsubsection{When \(\beta = 0\) (Cost Minimization Only)}
\label{SSSC:beta:eq:0}
Since all QCs have identical costs and teleportation is free, minimizing the total 
leasing cost requires minimizing the number of QCs used. The optimal strategy 
is to use the minimum number of QCs necessary to satisfy the capacity 
constraints. The minimum number of QCs needed is determined by both storage 
and execution capacity constraints:
\begin{equation}
k_{\min} = \max\left\{ \left\lceil \frac{|Q|}{s} \right\rceil, 
\left\lceil \frac{\max_{t} |\{q : \text{executing at } t\}|}{e} 
\right\rceil \right\},
\end{equation}
where the first term ensures sufficient storage capacity across all QCs and 
the second term ensures sufficient execution capacity for the peak execution 
demand. An optimal solution is to select any $k_{\min}$ QCs from $P$ (all 
choices are optimal since the QCs are identical), say QCs 
$\{p_1, p_2, \ldots, p_{k_{\min}}\}$, and set $r^{s}_{p} = r^{e}_{p} = 0$ 
for all $p \notin \{p_1, \ldots, p_{k_{\min}}\}$.

For the selected QCs $\{p_1, \ldots, p_{k_{\min}}\}$, the optimal values of 
$r^{s}_{p_i}$ and $r^{e}_{p_i}$ are not fixed in advance but depend on the 
specific assignment of qubits and gates to QCs and their schedule over time. 
In general, any feasible assignment of qubits to the $k_{\min}$ selected QCs 
and gate schedule is optimal when $\beta = 0$, provided it respects the 
capacity constraints \eqref{EQ:resource:usage:constraint1} and 
\eqref{EQ:resource:usage:constraint2}. Given such an assignment and schedule, 
the minimum sufficient leasing quantities are
\begin{align}
r^{s}_{p_i} &= \max_{t \in T} \sum_{q \in Q} x_{q,p_i,t}, 
\quad i \in \{1, \ldots, k_{\min}\}, \\
r^{e}_{p_i} &= \max_{t \in T} \sum_{q \in Q} z_{q,p_i,t}, 
\quad i \in \{1, \ldots, k_{\min}\},
\end{align}
i.e., the storage (respectively, execution) leasing at each QC $p_i$ equals 
the peak number of qubits stored (respectively, actively executing) at $p_i$ 
over all time slots, ensuring no over-leasing occurs. The optimal total cost 
is
\begin{align}
C^{*}_{\beta=0} = \sum_{i=1}^{k_{\min}} \left( c^s \cdot r^{s}_{p_i} 
+ c^e \cdot r^{e}_{p_i} \right) + \sum_{g \in G} c^{g},
\end{align}
where the leasing terms are minimized by choosing the assignment and schedule 
that minimize the peak storage and execution usage at each QC. Since all QCs 
are identical and $c^s, c^e$ are the same for all QCs, the total leasing cost 
satisfies
\begin{align}
\sum_{i=1}^{k_{\min}} \left( c^s \cdot r^{s}_{p_i} + c^e \cdot r^{e}_{p_i} 
\right) &\geq c^s \cdot |Q| \notag \\
& \quad + c^e \cdot \max_{t} 
|\{q : \text{executing at } t\}|,
\end{align}
with equality achievable when qubits and gates are distributed such that the 
peak storage and execution loads are spread as evenly as possible across the 
$k_{\min}$ QCs.

Intuitively, exactly $k_{\min}$ QCs are used for the following reason. Using 
fewer than $k_{\min}$ QCs violates capacity constraints and hence is 
infeasible. Using more than $k_{\min}$ QCs increases the total leasing cost 
without reducing it when $\beta = 0$, since the total storage and execution 
requirements are fixed by the circuit and do not decrease by spreading them 
across more QCs.

We now explain why multiple optimal solutions exist. Since all QCs are 
identical and the teleportation cost is zero, there are $\binom{|P|}{k_{\min}}$ 
optimal choices of which $k_{\min}$ QCs to use. For each such choice, any 
feasible assignment of qubits to the selected QCs and any feasible gate 
schedule achieves the same total gate execution cost $\sum_{g \in G} c^g$, 
and the total leasing cost depends only on the peak storage and execution 
loads induced by the assignment and schedule. Hence, the optimal solution 
is not unique, and any assignment and schedule that minimizes the peak loads 
at each selected QC is optimal.

 \subsubsection{When \(\beta > 0\) (Makespan Matters)}
\label{SSSC:beta:gt:0}
The optimal solution trades off leasing cost against makespan by distributing 
gates across multiple QCs to enable parallel execution. Since all QCs are 
identical and teleportation is free, the optimal number of active QCs $k^*$ is the 
solution to
\begin{align}
k^* &= \underset{k \in \{k_{\min}, \ldots, |P|\}}{\mathrm{argmin}}\, \left\{ 
\sum_{i=1}^{k} \left( c^s \cdot r^{s}_{p_i}(k) + c^e \cdot 
r^{e}_{p_i}(k) \right) \right. \notag \\
& \qquad \qquad \qquad \qquad   + \beta\, T_{\max}(k) \Biggr\},
\end{align}
where $r^{s}_{p_i}(k) = \max_{t \in T} \sum_{q} x_{q,p_i,t}$ and 
$r^{e}_{p_i}(k) = \max_{t \in T} \sum_{q} z_{q,p_i,t}$ are the optimal 
leasing quantities at QC $p_i$ under the best feasible assignment and 
schedule using $k$ QCs, and $T_{\max}(k)$ is the minimum achievable makespan 
using $k$ QCs. Both $r^{s}_{p_i}(k)$, $r^{e}_{p_i}(k)$, and $T_{\max}(k)$ 
depend on the specific gate assignment and qubit placement chosen for $k$ QCs, 
and need not correspond to an equal distribution of qubits or gates across QCs. 
In general, finding the optimal assignment and schedule for a given $k$ requires 
solving a scheduling subproblem, and $T_{\max}(k)$ is determined by both the 
circuit's critical path length and the available parallelism. Since the leasing 
cost and makespan trade off against each other as $k$ varies-- using more QCs 
can reduce makespan but increases total leasing cost-- the optimal $k^*$ is 
found by evaluating the above objective over all feasible $k \in 
\{k_{\min}, \ldots, |P|\}$.

In the general case in which the circuit has limited parallelism or 
heterogeneous gate dependencies, $T_{\max}(k)$ may not decrease with $k$ 
beyond a certain point (e.g., when the critical path dominates), and the 
optimal leasing quantities may be highly asymmetric across QCs due to uneven 
gate distributions. In such cases, the full ILP given in Section 
\ref{SSC:problem:formulation} must be solved to find both the optimal $k^*$ 
and the corresponding optimal assignment and schedule.

\subsubsection{Characterization of All Optimal Solutions}

A solution is optimal if and only if the following conditions hold: The solution satisfies \emph{minimal leasing}, wherein for each QC \(p\), \(r^{s}_{p} = \max_{t} \sum_{q} x_{q,p,t}\) and \(r^{e}_{p} = \max_{t} \sum_{q} z_{q,p,t}\), ensuring that no over-leasing occurs. The gate schedule achieves \emph{optimal makespan} \(T_{\max}^{*}\) for the chosen set of active QCs, representing the minimum possible execution time. Finally, \emph{optimal QC selection} ensures that the number of active QCs \(k\) minimizes the total cost as specified in the formula above.

Since the teleportation cost is zero, the specific assignment of qubits and gates to QCs is irrelevant, as long as capacity constraints are satisfied and makespan is minimized. 

Note that the assumption \(E(p', p) = 0\) we have made in Case 2 means that teleportation has zero \emph{cost}, but does not necessarily imply instantaneous teleportation. Two interpretations are possible: (i) \textbf{Zero cost, zero time} where teleportation is both free and instantaneous, allowing qubits to be teleported between QCs without affecting the makespan; in this case, the optimization focuses purely on minimizing leasing costs, while achieving the best possible makespan through parallelization. (ii) \textbf{Zero cost, non-zero time} assumes that teleportation is free, but takes time \(\tau_{\text{teleport}}(p', p) > 0, \, \forall p', p\), requiring the makespan calculation to account for teleportation delays as \(T_{\max} \geq T_{\max}^{\text{crit}} + \sum_{p',p} n_{\text{teleport}}(p',p) \cdot \tau_{\text{teleport}}(p',p)\), where \(n_{\text{teleport}}(p',p)\) is the number of qubit teleportations from $p'$ to $p$ in the solution. We have adopted interpretation (i): teleportation is instantaneous (\( \tau_{\text{teleport}}(p', p) = 0, \, \forall p', p\)). Under this assumption, there are exponentially many optimal solutions that differ only in which QCs are used and how work is distributed among them, but all achieve the same total cost. Note that if teleportation takes non-zero time, then even with zero cost, the problem becomes more constrained, as the optimizer must minimize the number of teleportations to reduce makespan, even though teleportation cost does not contribute to the objective function.

\subsubsection{Conditions for the Optimality of Load Balancing}
In general, load balancing may not be optimal due to the following reason. 
Load balancing across all \(|P|\) QCs minimizes the makespan, but 
\emph{increases the leasing cost} because more QCs must lease resources. 
When \(\beta\) is small (cost-dominant regime), concentrating the load on 
fewer QCs is cheaper.

A sufficient condition for load balancing across all \(|P|\) QCs to be 
optimal is that for all \(k \in \{1, \ldots, |P|-1\}\),
\begin{align}
\label{eq:load-balance-sufficient}
&\sum_{i=1}^{|P|} \left(c^s \cdot r^s_{p_i}(|P|) + c^e \cdot r^e_{p_i}(|P|)\right) 
+ \beta \cdot T_{\max}(|P|) \notag \\
&< \sum_{i=1}^{k} \left(c^s \cdot r^s_{p_i}(k) + c^e \cdot r^e_{p_i}(k)\right) 
+ \beta \cdot T_{\max}(k),
\end{align}
where $r^s_{p_i}(k)$ and $r^e_{p_i}(k)$ are the optimal peak storage and 
execution leasing quantities at QC $p_i$ under the best feasible assignment 
and schedule using $k$ QCs, and $T_{\max}(k)$ is the minimum achievable 
makespan using $k$ QCs. Both quantities depend on the specific circuit 
structure, gate dependencies, and qubit assignments, and must in general be 
evaluated by solving the scheduling subproblem for each $k$. This condition 
must be verified for each $k$, requiring $O(|P|)$ comparisons.

We now provide a cost formula for the special case in which the circuit is 
perfectly parallelizable, meaning its gates can be partitioned into $k$ 
independent groups of equal total duration executable simultaneously on $k$ 
QCs with no inter-QC dependencies. In this case, the cost of using $k$ QCs 
is given by
\begin{align}
\text{Cost}(k) &= \sum_{i=1}^{k}\left(c^s \cdot r^s_{p_i}(k) + 
c^e \cdot r^e_{p_i}(k)\right) \notag \\
&\quad + \sum_{g \in G} c^{g} + 
\beta \cdot \max\left(T_{\max}^{\text{crit}}, \frac{W_{\text{total}}}{k}\right),
\label{EQ:load:balancing:opt:cost:function}
\end{align}
where \(W_{\text{total}} = \sum_{g \in G} \tau_g\) is the total work and 
\(T_{\max}^{\text{crit}}\) is the critical path length. The $\max(\cdot)$ 
term reflects the fact that $T_{\max}(k)$ is bounded below by both the 
critical path length (due to sequential dependencies that cannot be 
parallelized) and $\frac{W_{\text{total}}}{k}$ (the average work per QC even 
with perfect load balancing). For circuits that do not satisfy the perfect 
parallelizability assumption, \eqref{EQ:load:balancing:opt:cost:function} 
provides only a lower bound on the true cost, and the full ILP given in 
Section~\ref{SSC:problem:formulation} must be solved to find the optimal 
solution.

\subsection{Case 3: Chain Topology with Sequential Gates}
\label{subsec:chain-sequential}

\subsubsection{Single-Qubit Sequential Gates} 
\label{SSSC:single:qubit:sequential:gates}
\hfill\\
\indent 
\textbf{Problem Setting:} The circuit is a linear chain of single-qubit gates 
\(g_1, g_2, \ldots, g_{|G|}\) on one qubit \(q\), with precedence constraints 
\(g_i \prec g_{i+1}\) for all $i \in \{1,\ldots, |G|-1\}$, and gate durations 
$\tau_g = 1$ for all $g \in G$. The QCs form an arbitrary network topology.

Note that this special case, while theoretically tractable, represents only a 
very limited class of quantum transformations (arbitrary single-qubit unitaries) 
and hence does not capture the complexity of realistic quantum circuits.

\textbf{Solution:} This problem reduces to finding a minimum-cost path in a 
time-expanded graph \(\mathcal{G}\), where each node \((p, i)\) represents the 
qubit being at QC \(p\) after executing gate \(g_i\), and each edge represents 
either staying at the current QC to execute the next gate or teleporting to another 
QC before executing it. The cost of a stay edge from \((p, i)\) to \((p, i+1)\) 
is \(c^{g_{i+1}}_p\) (gate execution cost at $p$). The cost of a teleport edge from 
\((p, i)\) to \((p', i+1)\) is \(E(p,p') + c^{g_{i+1}}_{p'}\) 
(teleportation cost plus gate execution cost at $p'$), plus the one-time leasing cost 
$c^s_{p'} + c^e_{p'}$ if $p'$ has not been used before.

\begin{theorem}
The shortest path in \(\mathcal{G}\) corresponds to an optimal solution to the 
special case of the JQLQCD problem in which the quantum circuit is a single-qubit 
sequential circuit with unit-duration gates.
\end{theorem}
\begin{IEEEproof}
Any feasible solution specifies a sequence of QC locations 
\((p_0, p_1, \ldots, p_{|G|})\) where \(p_i\) is the QC where gate \(g_i\) 
executes, which naturally corresponds to a path 
\((p_0, 0) \to (p_1, 1) \to \cdots \to (p_{|G|}, |G|)\) in \(\mathcal{G}\). 
The total cost of a solution is
\begin{equation}
\sum_{p \in P_{\text{used}}} (c^{s}_p + c^{e}_p) + \sum_{i=1}^{|G|} 
c^{g_i}_{p_i} + \sum_{i=1}^{|G|} \mathbb{1}_{\{p_i \neq p_{i-1}\}} \cdot 
E(p_{i-1}, p_i),
\end{equation}
where \(P_{\text{used}} = \{p_i : i \in \{0, 1, \ldots, |G|\}\}\) is the set 
of distinct QCs used, with leasing costs paid once per QC. Since \(\mathcal{G}\) 
is a DAG with non-negative edge costs, Dijkstra's algorithm (or DP) finds the 
shortest path, which corresponds to the minimum-cost solution when edge costs 
properly account for one-time leasing charges as described in Algorithm 
\ref{alg:singlequbit-path}. Any other solution corresponds to a longer path 
with higher cost.
\end{IEEEproof}

\begin{algorithm}[hbt]
\caption{Minimum-Cost Path for Sequential Single-Qubit Gate Execution}
\label{alg:singlequbit-path}
\begin{algorithmic}[1]
\State \textbf{Input:} Set of QCs \(P\), gates \(G = \{g_1, \ldots, g_{|G|}\}\), 
costs \(c^s_p, c^e_p, c^g_{p}, E(p,p')\)
\State \textbf{Output:} Minimum-cost execution path
\State Construct the DAG \(\mathcal{G} = (\mathcal{V}, \mathcal{E})\):
\State \quad Nodes: \(\mathcal{V} = \{(p, i) : p \in P,\ i \in \{0, 1, \ldots, 
|G|\}\}\)
\State \quad \quad where \((p, i)\) means the qubit is at QC \(p\) after 
executing gate \(g_i\) (\(i=0\) is the initial state)
\State Initialize edges \(\mathcal{E} = \emptyset\), leasing costs 
$L(p) = c^s_p + c^e_p$ for each $p \in P$
\For{\(i = 0\) to \(|G|-1\)}
    \For{each \((p, i) \in \mathcal{V}\)}
        \State \textbf{Stay edge:} Add \((p, i) \to (p, i+1)\) with cost 
        \(c^{g_{i+1}}_{p}\)
        \For{each \(p' \in P,\ p' \neq p\)}
            \State \textbf{Teleport edge:} Add \((p, i) \to (p', i+1)\) with cost 
            \(E(p,p') + c^{g_{i+1}}_{p'}\)
        \EndFor
    \EndFor
\EndFor
\State Apply Dijkstra's algorithm or DP on \(\mathcal{G}\), initializing each 
starting node \((p_{\text{init}}, 0)\) with cost \(c^s_{p_{\text{init}}} + 
c^e_{p_{\text{init}}}\), and adding the one-time leasing cost $c^s_{p'} + 
c^e_{p'}$ the first time any QC $p'$ is visited along a path
\State \Return Minimum-cost path and its total cost
\end{algorithmic}
\end{algorithm}

The nodes of \(\mathcal{G}\) are pairs \((p, i)\) with \(p \in P\) and 
\(i \in \{0, \ldots, |G|\}\), giving \(O(|P| \cdot |G|)\) nodes. Each node 
has at most \(|P|\) outgoing edges, giving \(O(|P|^2 \cdot |G|)\) edges in 
total. The one-time leasing costs are tracked along each path during the 
shortest-path computation. Using Dijkstra's algorithm with a Fibonacci heap, 
the complexity is \(O(|P|^2 \cdot |G| \cdot \log(|P| \cdot |G|))\), or 
\(O(|P|^2 \cdot |G|)\) using DP.

\subsubsection{Multi-Qubit Sequential Gates}
\label{SSSC:multi:qubit:sequential:gates}

\hfill\\
\indent 
\textbf{Problem Setting:} The circuit is a linear chain of gates 
\(g_1, g_2, \ldots, g_{|G|}\) with precedence constraints \(g_i \prec g_{i+1}\) 
for all $i \in \{1,\ldots, |G|-1\}$ and gate durations $\tau_g = 1$ for all 
$g \in G$. Each gate \(g_i\) operates on a subset of qubits 
\(O_{g_i} \subseteq Q\) with \(|O_{g_i}| \geq 1\). QCs form an arbitrary 
network topology. This is a generalization of the scenario considered in Section 
\ref{SSSC:single:qubit:sequential:gates} and encompasses realistic quantum 
circuits such as sequential quantum Fourier transforms (QFTs), layered variational quantum circuits, CNOT 
ladders, and quantum error correction (QEC) syndrome extraction circuits \cite{nielsen2010quantum}.

\textbf{Solution:} The problem still reduces to shortest path finding, but 
the state space must track the locations of all qubits involved in each gate. 
Algorithm \ref{alg:multiqubit-path} extends Algorithm \ref{alg:singlequbit-path} 
to handle multiple operand qubits per gate.

\begin{algorithm}
\caption{Minimum-Cost Path for Multi-Qubit Sequential Gate Execution}
\label{alg:multiqubit-path}
\begin{algorithmic}[1]
\State \textbf{Input:} Gates \(G = \{g_1, \ldots, g_{|G|}\}\), QCs \(P\), 
qubit operand sets \(O_{g_i}\), costs \(c^s_p, c^e_p, c^g_{p},  
E(p,p')\)
\State \textbf{Output:} Minimum-cost execution path
\State Construct DAG \(\mathcal{G} = (\mathcal{V}, \mathcal{E})\):
\State \quad Nodes: \(\mathcal{V} = \{(\sigma, i) : \sigma \in 
\text{Config}(O_{g_i}),\ i \in \{0, 1, \ldots, |G|\}\}\)
\State \quad \quad where \(\text{Config}(O_{g_i}) = \{\sigma : O_{g_i} \to P\}\) 
is the set of all mappings from qubits in \(O_{g_i}\) to QCs, with 
\(|\text{Config}(O_{g_i})| = |P|^{|O_{g_i}|}\), and \((\sigma, i)\) represents 
the qubit configuration after executing gate \(g_i\)
\State Initialize edges \(\mathcal{E} = \emptyset\)
\For{\(i = 0\) to \(|G|-1\)}
    \For{each state \((\sigma, i) \in \mathcal{V}\)}
        \For{each \(p \in P\) with \(|O_{g_{i+1}}| \leq \min(s_p, e_p)\) and 
        \(A_{p,g_{i+1}} = 1\)}
            \State Define \(\sigma': O_{g_{i+1}} \to P\) by \(\sigma'(q) = p\) 
            for all \(q \in O_{g_{i+1}}\)
            \State Compute teleportation cost: 
            \(\text{TeleportCost}(\sigma, \sigma') = \sum_{q \in O_{g_{i+1}}} 
            \mathbb{1}_{\{\sigma(q) \neq p\}} \cdot 
             E(\sigma(q), p)\)
            \State Compute incremental leasing cost: For each QC $p$, let 
            $n_{\text{prev}}(p)$ be the number of qubits leased at $p$ before 
            stage $i+1$ and $n_{\text{new}}(p)$ be the number of qubits in 
            $O_{g_{i+1}}$ assigned to $p$ under $\sigma'$; the incremental 
            leasing cost is:
            \(\text{LeaseCost} = \sum_{p \in P} 
            \max(0,\ n_{\text{new}}(p) - n_{\text{prev}}(p)) \cdot 
            (c^s_p + c^e_p)\)
            \State Compute edge cost: \(c = \text{LeaseCost} + 
            c^{g_{i+1}}_{p} + \text{TeleportCost}(\sigma, \sigma')\)
            \State Add edge \((\sigma, i) \to (\sigma', i+1)\) with cost 
            \(c\) to \(\mathcal{E}\)
        \EndFor
    \EndFor
\EndFor
\State Apply DP to find minimum-cost path from initial configuration 
\((\sigma_0, 0)\) with initial leasing cost 
\(\sum_{p \in P} n_{\text{init}}(p) \cdot (c^s_p + c^e_p)\), where 
$n_{\text{init}}(p)$ is the number of qubits in $O_{g_0}$ initially at $p$, 
to any final configuration \((\sigma_f, |G|)\)
\State \Return Minimum-cost path
\end{algorithmic}
\end{algorithm}

The nodes of \(\mathcal{G}\) are pairs \((\sigma, i)\) where 
\(\sigma \in \text{Config}(O_{g_i})\) and \(i \in \{0,\ldots,|G|\}\), giving 
\(O(|P|^{k_{\max}} \cdot |G|)\) nodes, where \(k_{\max} = \max_i |O_{g_i}|\). 
Each node has at most \(|P|\) outgoing edges (one per target QC for the next 
gate), giving \(O(|P|^{k_{\max}+1} \cdot |G|)\) edges. Using DP, the time 
complexity is \(O(|P|^{k_{\max}+1} \cdot |G|)\) and space complexity is 
\(O(|P|^{k_{\max}} \cdot |G|)\). For circuits with at most two-qubit gates 
(\(k_{\max} \leq 2\)), this gives \(O(|P|^3 \cdot |G|)\) time, which is 
polynomial and tractable for moderate-sized networks.

\begin{theorem}
The shortest path in \(\mathcal{G}\) corresponds to an optimal solution of 
the special case of the JQLQCD problem in which there is a sequential 
multi-qubit circuit with unit-duration gates.
\end{theorem}
\begin{IEEEproof}
Any feasible solution specifies a sequence of qubit configurations 
\((\sigma_0, \sigma_1, \ldots, \sigma_{|G|})\), which corresponds to a unique 
path in \(\mathcal{G}\). The total cost equals the sum of edge costs along 
the path, including incremental leasing costs accumulated as new qubits are 
assigned to QCs at each stage. Since \(\mathcal{G}\) is a DAG with 
non-negative edge costs, DP finds the optimal path.
\end{IEEEproof}

Some practical optimization strategies are as follows. Locality exploitation 
prunes configurations where qubits are unnecessarily dispersed if the circuit 
has qubit locality, meaning qubits in consecutive gates overlap significantly. 
Greedy initialization uses a greedy heuristic to identify a good initial path, 
then performs local search around this solution. Capacity-aware pruning 
eliminates configurations that violate capacity constraints early to reduce 
the state space.

\subsection{Case 4: Independent Subcircuits with Partitioned Resources}
\label{subsec:independent-subcircuits}
\textbf{Problem Setting:} The circuit decomposes into \(k\) independent 
subcircuits \(\{C_1, C_2, \ldots, C_k\}\) with no shared qubits. Different 
subcircuits \(C_1, \ldots, C_k\) operate on disjoint qubit sets 
\(Q_1, \ldots, Q_k\), respectively.

This problem represents a significant generalization beyond the \(k=1\) case, 
and we defer the full treatment of the \(k > 1\) case until future work. For 
completeness, we provide a decomposition-based approach, which is described in 
Algorithm \ref{alg:circuit_partitioning}. The decomposition step, which 
identifies the connected components of the qubit interaction graph 
$\mathcal{H} = (Q, E_{\mathcal{H}})$-- where $Q$ is the set of qubits and 
$E_{\mathcal{H}}$ is the set of edges with $(q, q') \in E_{\mathcal{H}}$ if 
there exists a gate in $G$ that operates on both qubits $q$ and $q'$-- has 
complexity $O(|Q| + |E_{\mathcal{H}}|)$ using depth-first search 
\cite{kleinberg2006algorithm}. Here, $|Q|$ is the number of qubits and 
$|E_{\mathcal{H}}| \leq |G|$ is the number of edges in the qubit interaction 
graph. The total complexity of Algorithm \ref{alg:circuit_partitioning} is 
$O(|Q| + |E_{\mathcal{H}}|)$ for the decomposition step, plus the sum of 
complexities for solving the JQLQCD problem for each subcircuit. This approach 
provides a feasible solution, but it may not be optimal if subcircuits can 
share QC resources efficiently.

\begin{algorithm}[h]
\caption{Decomposition-Based Approach for Circuit Partitioning and Assignment}
\label{alg:circuit_partitioning}
\begin{algorithmic}[1]
\State \textbf{Input:} Gate set $G$, qubit set $Q$, set of QCs \(P\)
\State \textbf{Output:} Assignment of circuit components to QCs
\State Construct the qubit interaction graph $\mathcal{H} = (Q, E_{\mathcal{H}})$, 
where $(q, q') \in E_{\mathcal{H}}$ if some gate in $G$ operates on both $q$ 
and $q'$
\State Identify connected components \(\{\tilde{C}_1, \tilde{C}_2, \ldots, 
\tilde{C}_m\}\) of $\mathcal{H}$ using depth-first search
\State \textbf{Note:} \(m\) is the number of connected components found (may 
differ from \(k\) in Section \ref{subsec:independent-subcircuits})
\State Assign each component \(\tilde{C}_i\) to QCs using a bin-packing 
heuristic \cite{johnson1974fast}, e.g., first-fit decreasing by resource demand
\For{each component \(\tilde{C}_i\)}
    \State Let \(P_i\) be the set of QCs assigned to component \(\tilde{C}_i\)
    \If{\(|P_i| = 1\)}
        \State Solve the subproblem for \(\tilde{C}_i\) using our method for 
        Case 1 (Section \ref{subsec:single-qc})
    \Else
        \State Solve the JQLQCD problem for circuit \(\tilde{C}_i\) and set 
        of QCs $P_i$
    \EndIf
\EndFor
\State \Return Combined solution from all subproblems
\end{algorithmic}
\end{algorithm}

\subsection{Case 5: Infinite Resources with Only Makespan Minimization}
\label{subsec:infinite-resources}

\textbf{Problem Setting:} For all \(p\) and $g$, \(s_p = e_p = \infty\), 
\(c^{s}_p = c^{e}_p = c^{g}_{p} = 0\), and \(\beta > 0\). Additionally, 
we assume that teleportation costs are zero, i.e., 
$E(p',p) = 0$ for all $p', p \in P$, and that qubit teleportation 
is instantaneous, i.e., $\tau_{\text{teleport}}(p',p) = 0$ for all $p', p \in P$. Under these 
assumptions, qubits can be freely teleported between QCs without incurring 
any cost or time overhead, and hence the objective function in \eqref{eq:obj} 
reduces to $\beta \cdot T_{\max}$, which is proportional to the makespan.

\textbf{Solution:} Since all resources are free, infinite, and teleportation is 
instantaneous, the only constraint on gate scheduling is the circuit's 
precedence structure. The problem therefore reduces to the classical problem 
of multiprocessor scheduling with precedence constraints, 
\(P|\text{prec}|C_{\max}\) \cite{Ullman1975, GareyJohnson1979, Lenstra1977}, 
described in Section \ref{SSC:NP:completenss:of:JQLQCD}.

The circuit's execution dependencies can be represented by a DAG in which 
there is a node corresponding to each gate \(g \in G\), and a directed edge 
\(g_1 \to g_2\) exists if gate \(g_1\) must complete before \(g_2\) can 
begin, i.e., \(g_1 \prec g_2\). This occurs when \(g_2\) operates on qubits 
modified by \(g_1\). The graph is acyclic because gates execute in temporal 
order without circular dependencies.

The optimal makespan is determined by the \emph{critical path} in the circuit 
DAG, which is the longest path from any source gate to any sink gate, 
accounting for gate durations:
\begin{equation}
\label{EQ:T:opt:max}
T_{\max}^{\text{opt}} = \max_{\text{paths } \pi \text{ in DAG}} \sum_{g \in \pi} 
\tau_g.
\end{equation}

Algorithm \ref{alg:critical-path} computes the critical path length using DP 
on the DAG and produces a gate assignment that achieves the optimal makespan. 
Since resources are infinite and free, and teleportation is instantaneous, gates 
can be distributed across as many QCs as needed to maximize parallelism: 
independent gates (those with no precedence relation) can be assigned to 
different QCs and executed simultaneously, while dependent gates must respect 
the ordering imposed by the precedence constraints. 
The complexity of Algorithm 
\ref{alg:critical-path} is \(O(|G| + |E|)\), where \(|E|\) is the number of 
precedence edges in the circuit DAG.

\begin{algorithm}
\caption{Critical Path Scheduling for Gate Assignment}
\label{alg:critical-path}
\begin{algorithmic}[1]
\State \textbf{Input:} Gates \(G\) with precedence constraints, durations 
\(\tau_g\), set of QCs \(P\)
\State \textbf{Output:} Gate assignment to QCs with optimal makespan
\State \textbf{Step 1: Compute critical path length}
\State Perform topological sort on gates \(G\)
\For{each gate \(g\) in topological order}
    \State Compute earliest start time:
    \State \quad \(L(g) = \max_{g' : g' \prec g} \big(L(g') + \tau_{g'}\big)\),
    \State \quad where \(L(g)\) is the earliest start time of gate \(g\)
\EndFor
\State Compute optimal makespan 
$T_{\max}^{\text{opt}} = \max_{g \in G} (L(g) + \tau_g)$
\State \textbf{Step 2: Assign gates to QCs}
\State Assign independent gates (those with no precedence relation between 
them) to distinct QCs to enable parallel execution, respecting precedence 
constraints; apply list scheduling \cite{Graham1969, Coffman1976} to 
produce a valid schedule achieving $T_{\max}^{\text{opt}}$
\State \Return Gate assignment and schedule
\end{algorithmic}
\end{algorithm}

For \(P|\text{prec}|C_{\max}\), finding the optimal makespan is strongly 
NP-hard in general \cite{Ullman1975}. However, the critical path provides a 
lower bound, and various polynomial-time algorithms achieve constant-factor 
approximations \cite{GrahamLawlerLenstraRinnooyKan1979OptimizationAA}. When 
\(|P| \geq |Q|\) (enough QCs to hold all qubits simultaneously) and teleportation 
is instantaneous, the critical path lower bound $T_{\max}^{\text{opt}}$ in 
\eqref{EQ:T:opt:max} is achievable, making the problem optimally solvable in 
polynomial time. This bound is tight when the circuit has sufficient 
parallelism and \(|P|\) is large enough.

\subsection{Case 6: Tree-Structured Circuit with Arbitrary QC Network}
\label{SSC:tree:structured:circuit}

\textbf{Problem Setting:} The circuit has a \emph{tree structure}, meaning 
that the gate dependency graph forms a tree. Since each node in a tree (except 
the root) has exactly one parent and each node has at most one child, a 
tree-structured gate dependency graph must in fact be a \emph{linear chain}: 
$g_1 \to g_2 \to \cdots \to g_{|G|}$, where each gate $g_i$ (except the last) 
has exactly one successor $g_{i+1}$, and each gate $g_i$ (except the first) 
has exactly one predecessor $g_{i-1}$. Note that a gate can still operate on 
multiple qubits (e.g., a two-qubit CNOT gate), but the dependency ordering 
among gates must form this linear chain structure.

We make the following assumptions. The leasing costs $c^s_p$ and $c^e_p$, 
gate execution costs $c^g_p$,  and teleportation 
costs $E(p',p)$ are arbitrary non-negative values. Qubit teleportation is 
instantaneous, i.e., $\tau_{\text{teleport}}(p',p) = 0$ for all $p', p \in P$, so that teleportation does not contribute to 
the makespan. All QCs have infinite storage and execution capacity: 
$s_p = e_p = \infty$ for all $p \in P$. The makespan weight is $\beta = 0$, 
so the objective function \eqref{eq:obj} reduces to minimizing only the 
leasing cost, gate execution cost, and teleportation cost. Under these assumptions, 
since all QCs have infinite capacity, the optimal leasing quantities are 
$r^s_p = \max_{t \in T} \sum_q x_{q,p,t}$ and $r^e_p = \max_{t \in T} \sum_q 
z_{q,p,t}$, i.e., the peak number of qubits stored and actively executing at 
each QC. 

The QC network can have an \emph{arbitrary topology}-- QCs can be connected 
in any graph structure (mesh, star, fully connected, etc.).

Note that Case 6 is equivalent to Case 3 (Section~\ref{subsec:chain-sequential}) 
under the additional assumption of $\beta = 0$ and instantaneous teleportation, 
since the tree structure reduces to a linear chain. The DP approach of 
Algorithm~\ref{alg:tree-dp} provides an alternative formulation that makes 
the optimal substructure explicit.

Algorithm \ref{alg:tree-dp} describes a DP approach that processes gates in 
post-order traversal (which coincides with reverse sequential order for a 
linear chain) to compute optimal costs bottom-up.

\begin{algorithm}
\caption{DP for Linear-Chain Circuit Optimization}
\label{alg:tree-dp}
\begin{algorithmic}[1]
\State \textbf{Input:} Circuit with linear-chain gate dependency 
$g_1 \to g_2 \to \cdots \to g_{|G|}$, set of QCs $P$ with leasing costs 
$c^s_p, c^e_p$, gate costs $c^g_p$, teleportation costs $E(p',p)$ 
for all $p, p' \in P$
\State \textbf{Output:} Optimal cost $C^*$ and gate assignment to QCs
\State \textbf{Initialization:}
\State Initialize $\text{Cost}(g, p)$ table for all gates $g$ and QCs $p$
\State \textbf{DP (processed in reverse order $g_{|G|}, g_{|G|-1}, \ldots, 
g_1$):}
\For{each gate $g$ in reverse sequential order}
    \For{each $p \in P$}
        \If{$g = g_{|G|}$ (last gate, no successor)}
            \State $\text{Cost}(g, p) \gets c^g_{p} + c^s_p + c^e_p$
        \Else
            \State Let $g'$ be the unique successor of $g$
            \State $\text{Cost}(g, p) \gets c^g_{p} + c^s_p + c^e_p + 
            \min_{p'' \in P} \left[ \text{Cost}(g', p'') + 
            E(p, p'') \right]$
        \EndIf
    \EndFor
\EndFor
\State \textbf{Compute optimal cost:}
\State $C^* \gets \min_{p \in P} \text{Cost}(g_1, p)$
\State \textbf{Reconstruction:}
\State Backtrack through the DP table to determine:
\State \quad - Which QC executes each gate
\State \quad - Which qubits are teleported between which QCs
\State \Return $C^*$ and optimal assignment
\end{algorithmic}
\end{algorithm}

\begin{theorem}
For circuits with linear-chain gate dependencies, infinite QC capacities, 
instantaneous qubit teleportation, $\beta = 0$, and an arbitrary QC network 
topology, Algorithm \ref{alg:tree-dp} computes the optimal solution.
\end{theorem}

\begin{IEEEproof}
We prove the result by induction on the gates processed in reverse sequential 
order.

\textit{Base case:} For the last gate $g_{|G|}$ with no successor, the 
minimum cost of executing $g_{|G|}$ at QC $p$ is
\begin{equation}
\text{Cost}(g_{|G|}, p) = c^g_{p} + c^s_p + c^e_p,
\end{equation}
which is trivially optimal since there are no subsequent gates.

\textit{Inductive step:} Assume that $\text{Cost}(g', p'')$ has been computed 
optimally for the unique successor $g'$ of gate $g$ and all QCs $p'' \in P$. 
We claim that
\begin{align}
\text{Cost}(g, p) &= c^g_{p} + c^s_p + c^e_p \notag \\
& \quad + \min_{p'' \in P} 
\left[\text{Cost}(g', p'') + \text{TeleportCost}(p, p'')\right]
\end{align}
computes the optimal cost for executing gate $g$ at QC $p$ and all subsequent 
gates, where $\text{TeleportCost}(p, p'') = E(p, p'')$.

Consider any feasible solution $\mathcal{S}$ for executing $g$ and all its 
successors. It must specify the QC $p$ where $g$ executes and the QC $p''$ 
where $g'$ executes. The total cost is:
\begin{equation}
C(\mathcal{S}) = c^g_p + c^s_p + c^e_p + C(\mathcal{S}_{g'}) + 
\text{TeleportCost}(p, p''),
\end{equation}
where $\mathcal{S}_{g'}$ is the subsolution for $g'$ and all its successors. 
By the induction hypothesis, $C(\mathcal{S}_{g'}) \geq \text{Cost}(g', p'')$. 
Therefore,
\begin{align*}
C(\mathcal{S}) &\geq c^g_p + c^s_p + c^e_p + \text{Cost}(g', p'') + 
\text{TeleportCost}(p, p'') \\
&\geq c^g_p + c^s_p + c^e_p \\
& \quad + \min_{p'' \in P}\left[\text{Cost}(g', p'') + 
\text{TeleportCost}(p, p'')\right] \\
&= \text{Cost}(g, p).
\end{align*}
Equality holds when $\mathcal{S}_{g'}$ is optimal and $p''$ is chosen 
optimally. The global optimal cost is $C^* = \min_{p \in P} 
\text{Cost}(g_1, p)$.
\end{IEEEproof}

We now analyze the time and space complexity of Algorithm \ref{alg:tree-dp}.
\begin{proposition}
Under the assumption that each gate operates on at most $k_{\max} \leq 2$ 
qubits, Algorithm \ref{alg:tree-dp} has time complexity $O(|G| \cdot |P|^2)$ 
and space complexity $O(|G| \cdot |P|)$.
\end{proposition}
\begin{IEEEproof}
The DP table $\text{Cost}(g, p)$ has $|G| \cdot |P|$ entries, giving space 
complexity $O(|G| \cdot |P|)$. For each gate $g$ and each QC $p$, computing 
$\text{Cost}(g, p)$ requires minimizing over all $p'' \in P$, taking $O(|P|)$ 
time. Summing over all gates and QCs gives total time complexity 
$O(|G| \cdot |P|^2)$.
\end{IEEEproof}

\section{Greedy Algorithm}
\label{SC:greedy:algorithm}
The special cases in Section \ref{sec:special-cases} cover only a narrow class
of practical scenarios. Real-world quantum circuits-- including QFTs, variational quantum eigensolvers (VQEs), and
QEC circuits \cite{nielsen2010quantum}-- exhibit DAG structures where gates have multiple predecessors and
successors, violating the tree property required by Case 6 (Section
\ref{SSC:tree:structured:circuit}). For such general circuits, exact solution of
the ILP formulation of the JQLQCD problem (Section \ref{SSC:problem:formulation})
is computationally prohibitive for large instances, since the problem is
NP-complete (see Section \ref{sec:complexity}). Hence, in this section, we
present a greedy heuristic for solving the general JQLQCD problem.

\subsection{Algorithm Design}
\subsubsection{QC Scoring Function}
The key idea is to iteratively select QCs based on a composite score that
balances leasing costs and communication overhead, and then allocate qubits
until capacity is exhausted. For each QC \(p \in P\), we define a
\emph{selection score}:
\begin{equation}
\begin{split}
\text{Score}(p) = \,&\alpha_{\text{lease}} \cdot
\frac{c^{s}_p + c^{e}_p}{\max_{p' \in P}(c^{s}_{p'} + c^{e}_{p'})} \\
+\,&\alpha_{\text{comm}} \cdot
\frac{\text{AvgCommCost}(p)}{\max_{p' \in P}\text{AvgCommCost}(p')},
\end{split}
\label{eq:qc-score}
\end{equation}
where \(\alpha_{\text{lease}}, \alpha_{\text{comm}} \ge 0\) are weights with
\(\alpha_{\text{lease}} + \alpha_{\text{comm}} = 1\), and
\(\text{AvgCommCost}(p)\) measures the average communication cost from
\(p\) to other QCs:
\begin{equation}
\label{eq:avg-comm-cost}
\text{AvgCommCost}(p) = \frac{1}{|P|-1}
\sum_{p' \in P \setminus \{p\}} E(p,p').
\end{equation}
Lower scores indicate more desirable QCs. The weights allow tuning based on
the problem requirements: if circuits require frequent qubit teleportation, then
high \(\alpha_{\text{comm}}\) is used to prioritize well-connected QCs;
otherwise, high \(\alpha_{\text{lease}}\) is used to prioritize cheap QCs.
Algorithm \ref{alg:greedy-qc-selection} describes this greedy selection
strategy; it iteratively chooses QCs in increasing order of score and assigns
qubits to QCs based on the gate execution costs and affinity, as we now explain.

\begin{algorithm}[H]
\caption{Greedy QC Selection and Qubit Allocation}
\label{alg:greedy-qc-selection}
\begin{algorithmic}[1]
\State \textbf{Input:} Circuit with qubits \(Q\), gates \(G\), QCs \(P\),
capacities \(s_p, e_p\), costs
\State \textbf{Output:} Initial qubit-to-QC assignment
\State Initialize: \(P_{\text{avail}} \leftarrow P\),
\(Q_{\text{unassigned}} \leftarrow Q\)
\State Compute \(\text{Score}(p)\) for all \(p \in P\) using \eqref{eq:qc-score}
\State Sort QCs in ascending order of score: \(p_1, p_2, \ldots, p_{|P|}\)
\For{each QC \(p_i\) in sorted order}
    \State \(n_{\text{assign}} \leftarrow \min(s_{p_i}, e_{p_i},
    |Q_{\text{unassigned}}|)\)
    \State Select \(n_{\text{assign}}\) qubits from \(Q_{\text{unassigned}}\)
    one at a time greedily, i.e., in increasing order of
    $\text{QubitScore}(q, p_i)$ given in \eqref{eq:qubit-score},
    recomputing scores after each selection since
    $\text{Affinity}(q, p_i)$ changes as qubits are assigned
    \State Assign selected qubits to \(p_i\)
    \State \(Q_{\text{unassigned}} \leftarrow Q_{\text{unassigned}}
    \setminus \{\text{selected qubits}\}\)
    \If{\(Q_{\text{unassigned}} = \emptyset\)}
        \State \textbf{break}
    \EndIf
\EndFor
\State \textbf{return} qubit-to-QC assignment
\end{algorithmic}
\end{algorithm}

\subsubsection{Qubit Selection Strategy}
While assigning qubits to a selected QC \(p\), we use a secondary greedy
criterion as follows. For each unassigned qubit \(q\), we compute
\begin{equation}
\label{eq:qubit-score}
\text{QubitScore}(q, p) = \sum_{g : q \in O_g}
\frac{c^{g}_{p}}{\min_{p' \in P} c^{g}_{p'}}
- \lambda \cdot \text{Affinity}(q, p),
\end{equation}
where \(\text{Affinity}(q, p)\) equals the fraction of gate operands involving
\(q\) that are already assigned to \(p\):
\begin{equation}
\label{eq:affinity}
\text{Affinity}(q, p) = \sum_{g : q \in O_g} \frac{|O_g \cap Q_p|}{|O_g|},
\end{equation}
with \(Q_p\) being the set of qubits already assigned to \(p\) and
\(\lambda \ge 0\) a weight parameter. Lower qubit scores are preferred.
The first term encourages assignment of qubits to a QC \(p\) where their
gates execute cheaply. The second term (subtracted) rewards co-location of
qubits that interact frequently: a higher \(\text{Affinity}(q,p)\) lowers
the score, making qubit \(q\) more attractive to assign to \(p\), thereby
reducing inter-QC communication. After each qubit is assigned to \(p_i\),
the \(\text{Affinity}(q, p_i)\) values change for the remaining unassigned
qubits, so \(\text{QubitScore}(q, p_i)\) must be recomputed before each
subsequent selection.

\subsection{Complexity Analysis}
Computing the communication costs $\text{AvgCommCost}(p)$ and scores
$\text{Score}(p)$ for all QCs $p \in P$ requires \(O(|P|^2)\) time, and
sorting QCs takes \(O(|P| \log |P|)\) time. For the qubit assignment step,
consider the processing of a single QC $p_i$. Up to
$n_{\text{assign}} \leq |Q|$ qubits are selected one at a time. Before each
selection, $\text{QubitScore}(q, p_i)$ must be recomputed for each remaining
unassigned qubit $q \in Q_{\text{unassigned}}$, since
$\text{Affinity}(q, p_i)$ changes after each qubit is assigned to $p_i$.
Recomputing $\text{QubitScore}(q, p_i)$ for a single qubit $q$ requires
$O(|G|)$ time (summing over all gates involving $q$). Since there are at most
$|Q|$ unassigned qubits and up to $|Q|$ selections per QC, the qubit score
evaluation for a single QC $p_i$ takes $O(|Q|^2 \cdot |G|)$ time in the
worst case. Summing over all $|P|$ QCs, the total time for qubit score
evaluation is $O(|P| \cdot |Q|^2 \cdot |G|)$. Hence, the overall time
complexity of the greedy algorithm is
\begin{equation}
O\!\left(|P|^{2} + |P| \cdot |Q|^{2} \cdot |G|\right). \label{eq:greedy-complexity}
\end{equation}

\subsection{Refinement via Local Search}
The greedy assignment obtained using Algorithm \ref{alg:greedy-qc-selection}
can be refined using local search, which is described in Algorithm
\ref{alg:local-search}.

\begin{algorithm}[hbt]
\caption{Local Search Refinement}
\label{alg:local-search}
\begin{algorithmic}[1]
\State \textbf{Input:} Assignment obtained using greedy algorithm
(Algorithm \ref{alg:greedy-qc-selection})
\State \textbf{Output:} Refined assignment
\State Compute total cost \(C_{\text{current}}\) of initial assignment
\Repeat
    \State \(C_{\text{best}} \leftarrow C_{\text{current}}\)
    \For{each qubit \(q\) and QC pair \((p, p')\) where \(q\) is
    assigned to \(p\)}
        \If{\(p'\) has available capacity and \(A_{p',g} = 1\) for all
        gates \(g\) with \(q \in O_g\)}
            \State Tentatively reassign \(q\) from \(p\) to \(p'\)
            \State Compute new cost \(C_{\text{new}}\)
            \If{\(C_{\text{new}} < C_{\text{best}}\)}
                \State \(C_{\text{best}} \leftarrow C_{\text{new}}\)
                \State Record this move as best improvement
            \EndIf
            \State Undo tentative reassignment
        \EndIf
    \EndFor
    \If{best improvement found}
        \State Apply best move
        \State \(C_{\text{current}} \leftarrow C_{\text{best}}\)
    \EndIf
\Until{no improving move found or iteration limit reached}
\State \textbf{return} refined assignment
\end{algorithmic}
\end{algorithm}

\section{Numerical Results}
\label{sec:num}
In this section, we evaluate the performance of the proposed greedy algorithm 
(Algorithm~\ref{alg:greedy-qc-selection}) through extensive numerical 
experiments. We compare the greedy approach with simulated annealing (SA) 
\cite{kleinberg2006algorithm}, a metaheuristic that can find near-optimal 
solutions to the ILP formulation, and polynomial-time algorithms that compute 
optimal solutions for special cases (Section \ref{sec:special-cases}). Our 
evaluation focuses on two key metrics: solution quality measured by total cost 
as defined by the objective function (see \eqref{eq:obj}), and computational 
efficiency measured by execution time in seconds.

\subsection{Experimental Setup}

\begin{table}[!ht]
\caption{The table provides the values/ distribution of different parameters 
used to obtain the numerical results.}
\label{table:instance_params}
\centering
\scriptsize
\begin{tabular}{|l|l|}
\hline
\rowcolor[HTML]{EFEFEF}
\textbf{Parameter} & \textbf{Values/ Distribution} \\ \hline
Number of qubits & $|Q| \in \{5, 10, 15, 20, 25, 30\}$ \\ \hline
Number of gates & $|G| \in \{10, 25, 50, 100, 150, 200\}$ \\ \hline
Gate types & Single-qubit ($60 \, \%$), two-qubit ($40 \, \%$) \\ \hline
Circuit topology & Random DAG, avg. degree 2.5 \\ \hline
Gate durations & Uniform$[1, 5]$ time units \\ \hline
Number of QCs & $|P| \in \{2, 3, 4, 5, 6, 8, 10\}$ \\ \hline
Storage capacity & $s_p \sim$ Uniform$\{5,\ldots, 15\}$ qubits \\ \hline
Execution capacity & $e_p \sim$ Uniform$\{3, \ldots, 8\}$ qubits \\ \hline
Network topology & Random graph, edge prob. 0.4 \\ \hline
Teleportation costs & $E(p', p) \sim$ Uniform$(10, 30)$ \\
& for connected QCs \\ \hline
Leasing costs & $c^s_p, c^e_p \sim$ Uniform$(1, 5)$ \\ \hline
Gate exec. costs & $c^g_{p} \sim$ Uniform$(0.5, 2.0)$ \\ \hline
Gate availability & $A_{p,g} = 1$ w.p. 0.8 ($80 \, \%$ avg.) \\ \hline
Makespan weight & $\beta \in \{0, 0.1, 0.5, 1.0, 2.0, 5.0\}$ \\ \hline
\end{tabular}
\end{table}

We generated diverse instances of the JQLQCD problem, with the characteristics 
shown in Table \ref{table:instance_params}, designed to represent realistic 
quantum computing scenarios. We used quantum circuits with varying numbers of 
qubits and gates, with gate types distributed as $60 \, \%$ single-qubit and 
$40 \, \%$ two-qubit gates, organized in a random DAG topology with average 
degree $2.5$, and gate durations uniformly distributed between $1$ and $5$ 
time units. The table also specifies various parameters of the QC network, 
including the number of QCs, storage and execution capacities per QC, network 
topology, which is a random graph with edge probability $0.4$, and various 
cost distributions for teleportation, leasing, and gate execution. 
The gate availability is set such that each gate can execute on approximately 
$80 \, \%$ of the QCs. We vary the makespan weight $\beta$ to study the 
trade-off between cost and makespan. For each parameter configuration, we 
report the results obtained by averaging over $20$ random instances.

\subsection{Algorithms}
We evaluated and compared the performance of the following algorithms.

\subsubsection{Greedy Algorithm} It was implemented as described in Algorithm 
\ref{alg:greedy-qc-selection} with parameters $\alpha_{\text{lease}} = 0.5$ 
and $\alpha_{\text{comm}} = 0.5$ (balanced weighting), $\lambda = 2.0$ 
(moderate affinity weight), and $100$ local search iterations 
(Algorithm~\ref{alg:local-search}).

\subsubsection{Simulated Annealing \cite{kleinberg2006algorithm}} The parameter
values used were as follows: initial temperature $T_0$ set to $5 \, \%$ of the
cost of the initial solution, geometric cooling schedule
$T_{k+1} = 0.95 \cdot T_k$, $50$ iterations per temperature, termination at
temperature $T < 0.01$ or $10,000$ total iterations, and neighborhood moves
consisting of random reassignment of one qubit to a different QC, the swapping
of two qubits, or the re-placement of a single gate at a randomly chosen
feasible QC. SA is initialized from the refined greedy solution. An adaptive
$T_0$ is used because a fixed initial temperature far above the cost scale of an
instance causes nearly every candidate move to be accepted, which reduces the
search to a random walk. The gate re-placement move is included because
$u_{g,p}$ is an independent decision variable of the ILP; with qubit moves
alone, SA cannot trade gate execution cost against makespan, which is precisely
the trade-off that Case 2 turns on.

\subsubsection{Optimal Algorithms} For the special cases studied in Section 
\ref{sec:special-cases}, we used topological sort plus list scheduling for 
Case 1 (Section \ref{subsec:single-qc}), optimal $k$ selection via exhaustive 
search over $k \in \{1, \ldots, |P|\}$ for Case 2 (Section 
\ref{subsec:homogeneous-zero-cost}), and Dijkstra's algorithm on the 
time-expanded graph for Case 3 (Section \ref{subsec:chain-sequential}).

All algorithms were implemented in Python 3.9 and executed on a machine with
an Intel Core i5-1335U CPU (10 cores, up to 4.6 GHz) and 16 GB RAM, running Linux.

\subsection{Results for General Problem Instances}

\subsubsection{Performance vs. Circuit Size}
Fig. \ref{fig:cost_vs_gates} shows the total cost as a function of the number of
gates $|G|$, with $|Q| = 20$ qubits, $|P| = 5$ QCs, and $\beta = 1.0$. Both
algorithms exhibit approximately linear cost growth with circuit size. The two
algorithms are nearly indistinguishable for small circuits ($|G| \leq 50$), and
the greedy algorithm's cost exceeds that of SA by roughly $10-15 \, \%$ for the
largest circuits tested. This widening is expected: as $|G|$ grows, the number
of gate-placement decisions grows with it, and the additional decisions give SA
more opportunities to improve upon the greedy solution from which it starts.
The total cost increases with $|G|$ for two primary reasons. First, the gate
execution cost component $\sum_{g \in G} \sum_{p \in P} c^g_p u_{g,p}$ in
\eqref{eq:obj} grows linearly with $|G|$. Second, more gates require additional
qubit teleportations to satisfy the operand co-location requirement
\eqref{EQ:operands:present}, increasing the teleportation cost term. The
makespan term $\beta T_{\max}$ also grows for circuits with longer critical
paths.

\begin{figure}[htbp]
\centering
\includegraphics[width=0.4\textwidth]{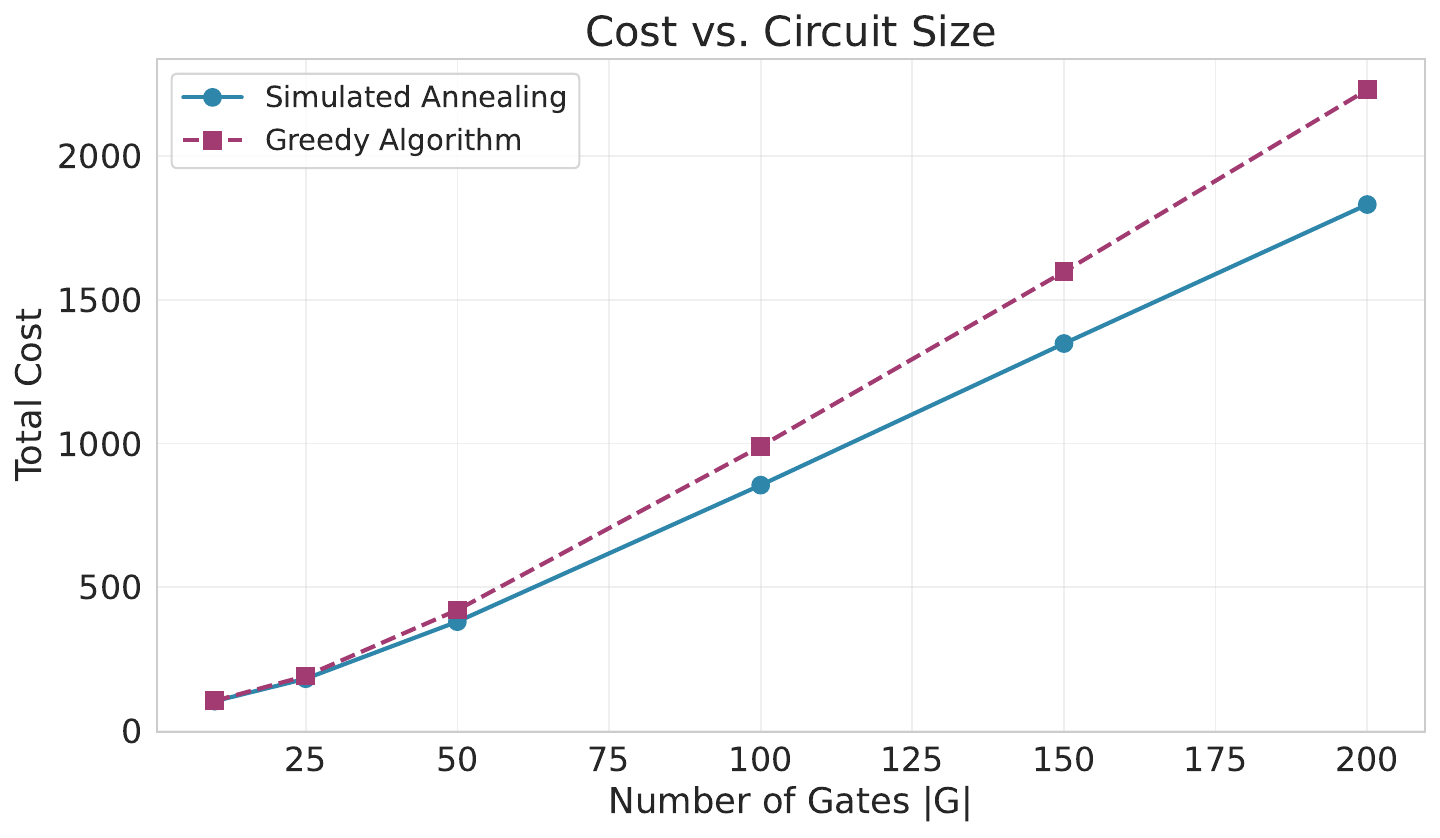}
\caption{The figure shows the total cost versus number of gates for the greedy 
algorithm and SA.}
\label{fig:cost_vs_gates}
\end{figure}
Fig. \ref{fig:time_vs_gates} compares the execution times for the same problem
instances. The greedy algorithm is consistently $10-20\times$ faster than SA
across all circuit sizes. For the largest instances ($|G| = 200$), the greedy
algorithm completes in under $1$ second, while SA requires approximately
$10$ seconds.

\begin{figure}[htbp]
\centering
\includegraphics[width=0.4\textwidth]{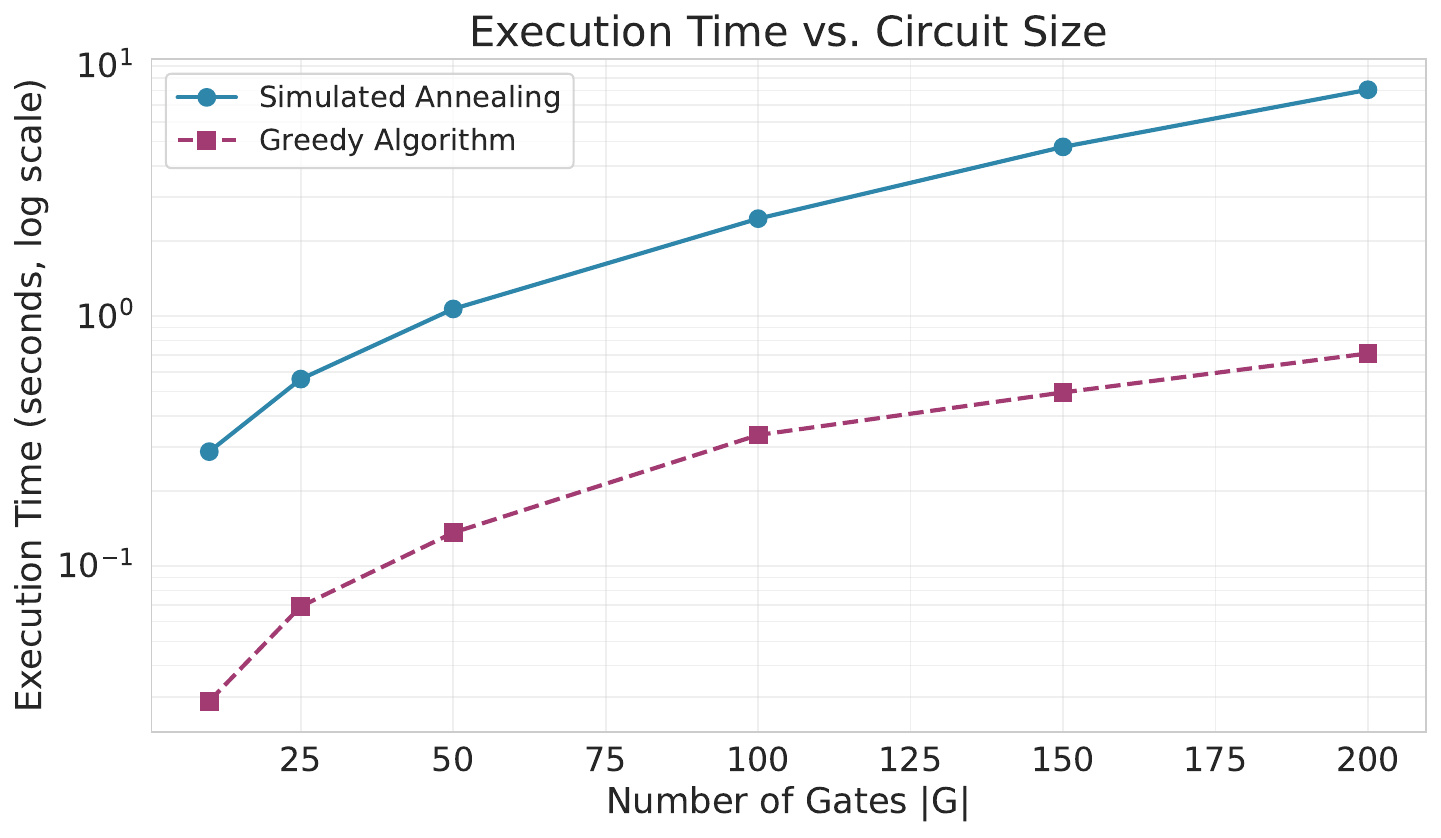}
\caption{The figure shows the execution time versus number of gates for the 
greedy algorithm and SA.}
\label{fig:time_vs_gates}
\end{figure}

\subsubsection{Performance vs. Number of QCs}
Fig. \ref{fig:cost_vs_qcs} shows how the solution quality varies with the number
of available QCs, with $|Q| = 15$ qubits, $|G| = 75$ gates, and $\beta = 1.0$.
Beyond $|P| = 3$, the cost decreases in the number of QCs for both algorithms,
falling by roughly $20 \, \%$ between $|P| = 3$ and $|P| = 10$: additional QCs
provide more opportunities for parallel gate execution and hence a lower
makespan. The non-monotonicity at $|P| = 3$ arises because the marginal QC is
initially used to relieve a capacity bottleneck rather than to expose
parallelism, so it adds leasing and teleportation cost before it begins to
reduce the makespan. The greedy algorithm's gap to SA narrows from about
$10 \, \%$ at small $|P|$ to under $10 \, \%$ at $|P| = 10$.

\begin{figure}[htbp]
\centering
\includegraphics[width=0.4\textwidth]{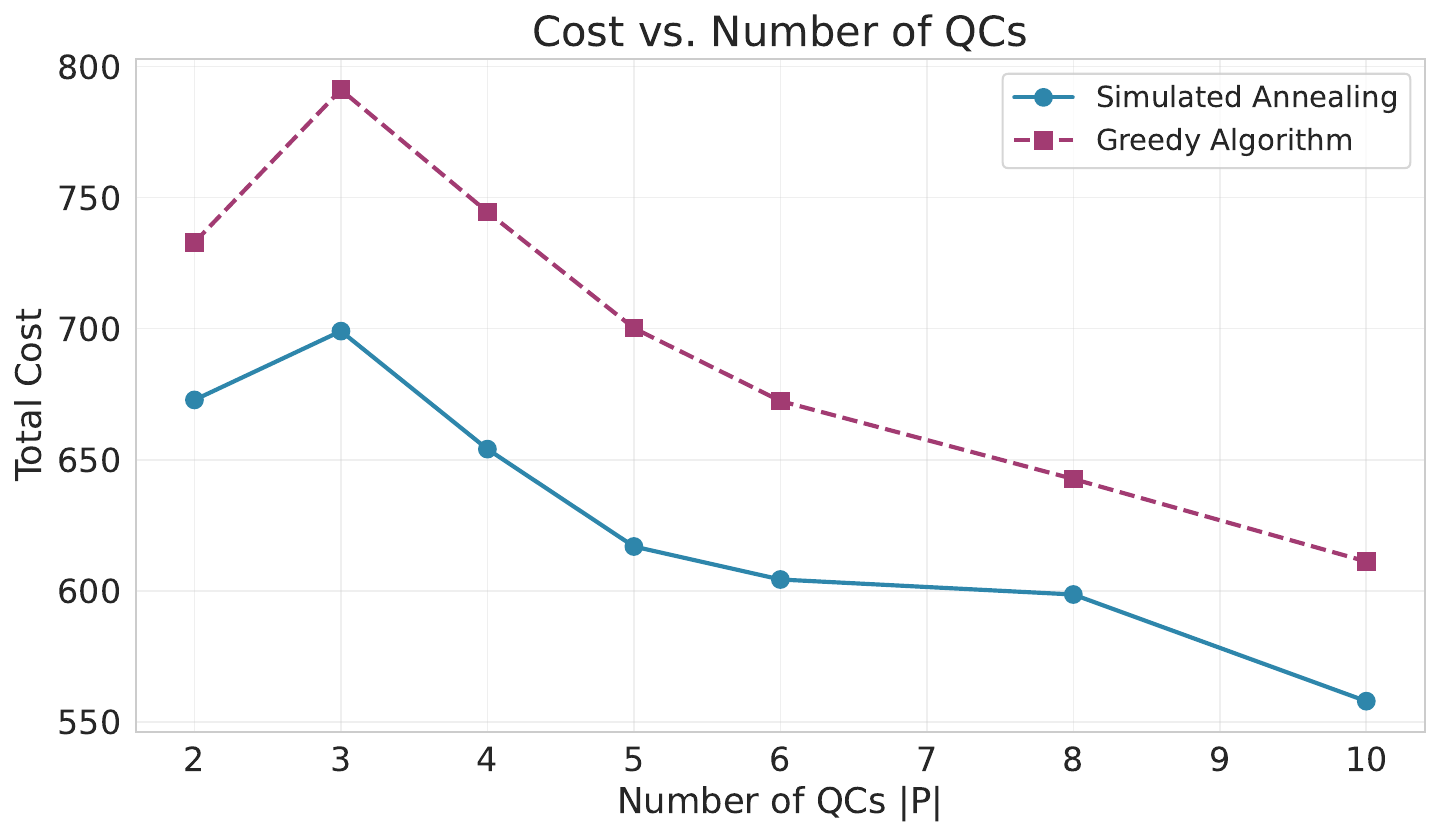}
\caption{The figure shows the total cost versus number of QCs for the greedy 
algorithm and SA.}
\label{fig:cost_vs_qcs}
\end{figure}

Fig. \ref{fig:time_vs_qcs} compares the execution times. SA's execution time
grows steadily with $|P|$, since enlarging the QC set enlarges its search space,
while the greedy algorithm's execution time remains below $0.2$ seconds
throughout. The resulting speedup ranges from about $40\times$ at $|P| = 2$ to
about $8\times$ at $|P| = 10$.

\begin{figure}[htbp]
\centering
\includegraphics[width=0.4\textwidth]{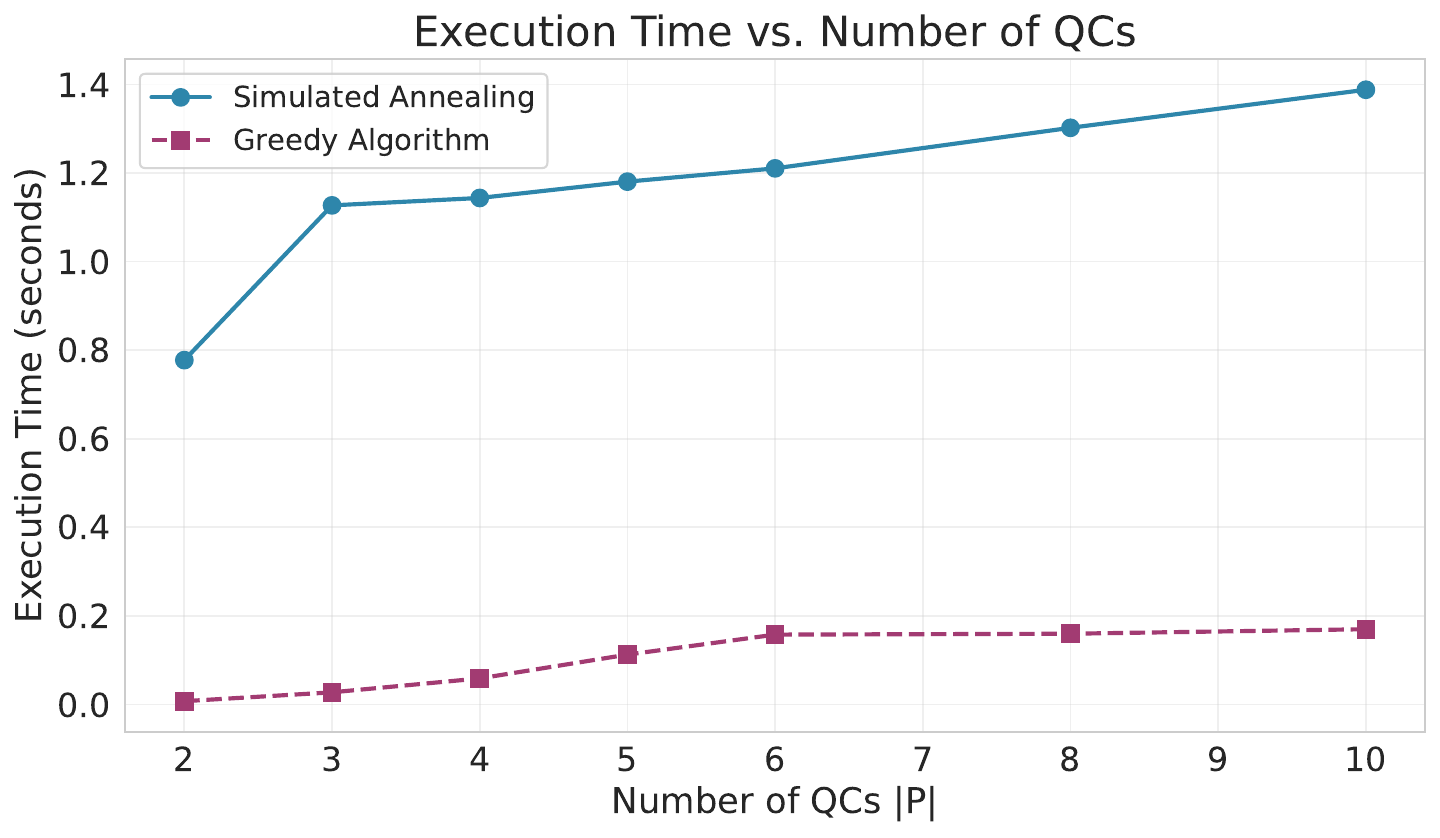}
\caption{The figure shows the execution time versus number of QCs for the 
greedy algorithm and SA.}
\label{fig:time_vs_qcs}
\end{figure}

\subsubsection{Performance vs. Makespan Weight $\beta$}
Fig. \ref{fig:cost_vs_beta} shows the total cost versus $\beta$, with
$|Q| = 15$, $|G| = 75$, and $|P| = 5$. Both costs grow roughly linearly in
$\beta$, since the makespan enters \eqref{eq:obj} linearly and both algorithms
respond by using more QCs and leasing more qubits to shorten the makespan, which
in turn raises the leasing and teleportation costs. The greedy algorithm stays
within about $10 \, \%$ of SA at $\beta = 0$, widening to roughly $8-10 \, \%$
at $\beta = 5$; the greedy algorithm is therefore most competitive in the
cost-dominated regime.

\begin{figure}[htbp]
\centering
\includegraphics[width=0.4\textwidth]{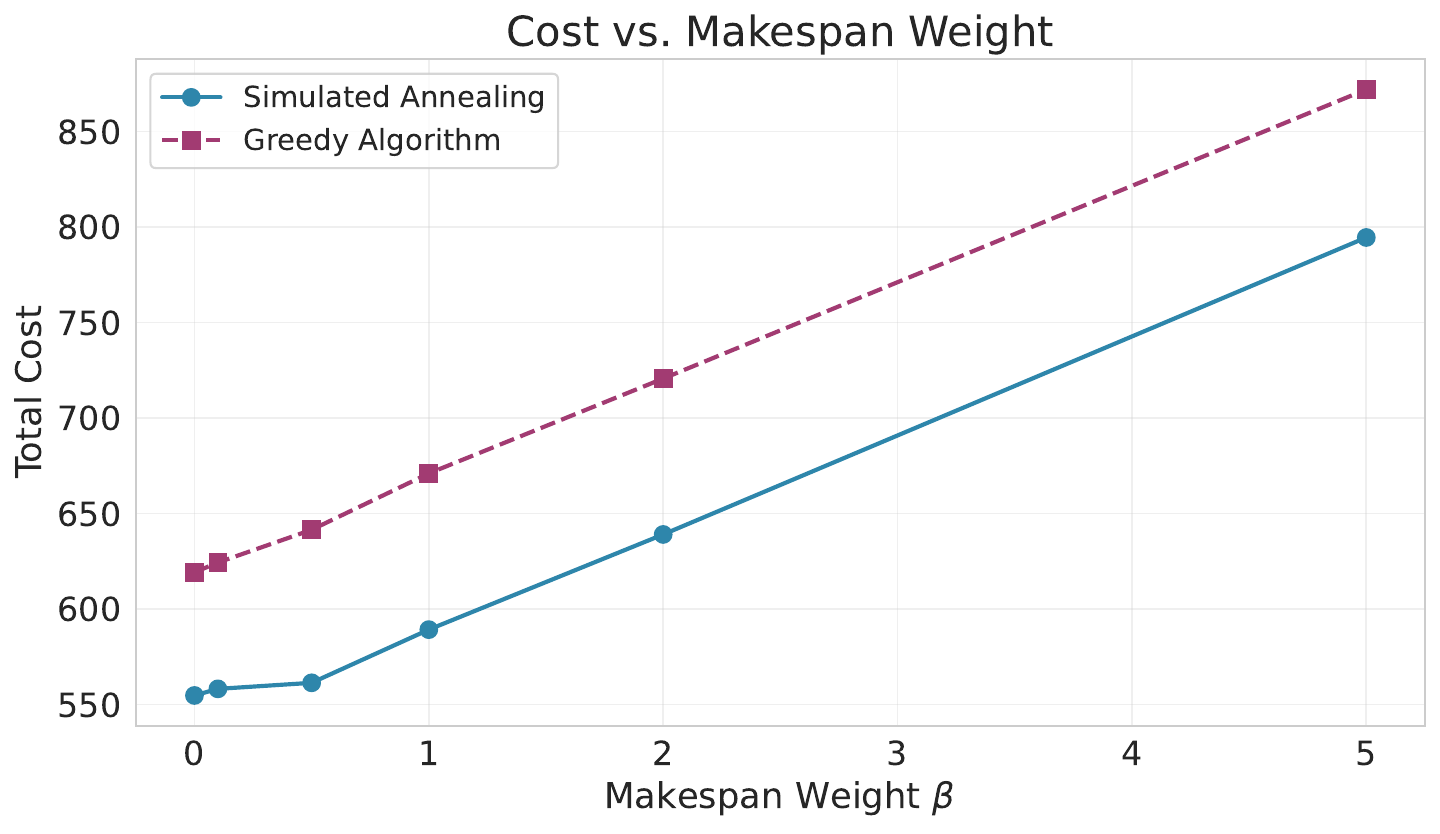}
\caption{The figure shows the total cost versus makespan weight $\beta$ for 
the greedy algorithm and SA.}
\label{fig:cost_vs_beta}
\end{figure}

Fig. \ref{fig:time_vs_beta} shows the execution times for the same instances.
Both algorithms maintain nearly constant execution times regardless of $\beta$,
as the running time is governed by the problem size rather than by the weighting
of the objective. The speedup ratio remains near $8\times$ across all values of
$\beta$.

\begin{figure}[htbp]
\centering
\includegraphics[width=0.4\textwidth]{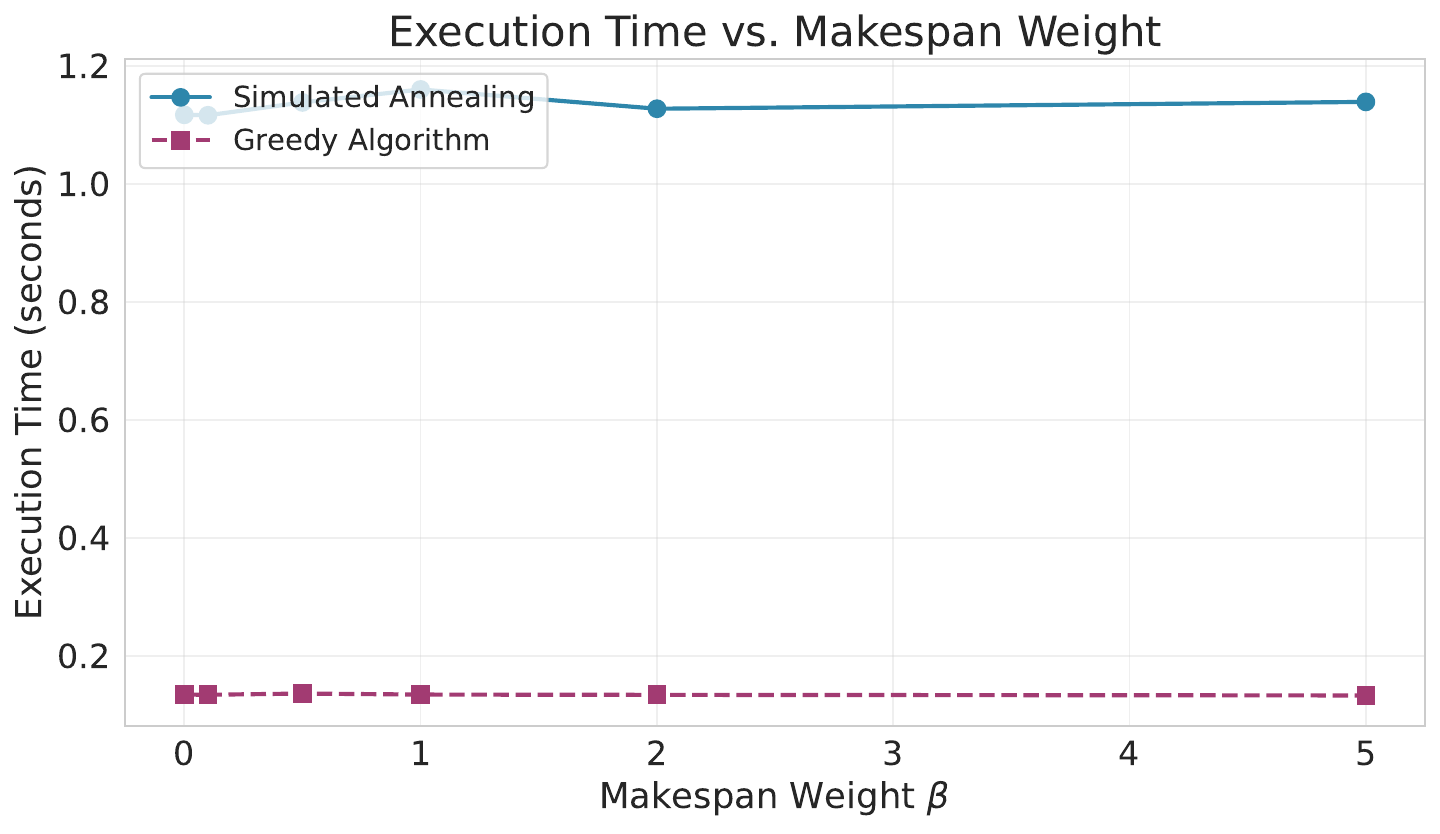}
\caption{The figure shows the execution time versus makespan weight $\beta$ 
for the greedy algorithm and SA.}
\label{fig:time_vs_beta}
\end{figure}

\subsection{Results for the Special Cases of Section \ref{sec:special-cases}}

\subsubsection{Case 1 (Section \ref{subsec:single-qc})}
We compare the greedy algorithm, SA, and the optimal algorithm (topological sort
plus list scheduling on $p_0$) for instances with one unlimited-capacity QC and
$3$ limited-capacity QCs. The instances are drawn so that Condition S2 of
Section \ref{SSSC:sufficient:conditions:for:optimality:of:centralized:solution}
holds -- $p_0$ has the cheapest leasing and gate execution costs, and
teleportation out of $p_0$ is prohibitive relative to the achievable savings --
so centralized execution at $p_0$ is provably optimal and the optimal curve is a
true lower bound. Fig. \ref{fig:case1_cost} shows that all three algorithms
achieve essentially the same cost across all circuit sizes, growing linearly
from about $60$ at $|G| = 25$ to about $145$ at $|G| = 100$. This is the
expected behavior in this regime: because teleportation out of $p_0$ cannot pay
for itself, every qubit is best kept at $p_0$, and both heuristics discover this
centralized solution. The practical implication is that when the sufficient
conditions of
Section \ref{SSSC:sufficient:conditions:for:optimality:of:centralized:solution}
hold, the greedy algorithm attains optimality rather than merely approaching it.
Fig. \ref{fig:case1_time} shows that the optimal algorithm is two to three
orders of magnitude faster than the two heuristics, since it constructs the
centralized solution directly instead of searching; the greedy algorithm is in
turn roughly an order of magnitude faster than SA.

\begin{figure}[htbp]
\centering
\includegraphics[width=0.4\textwidth]{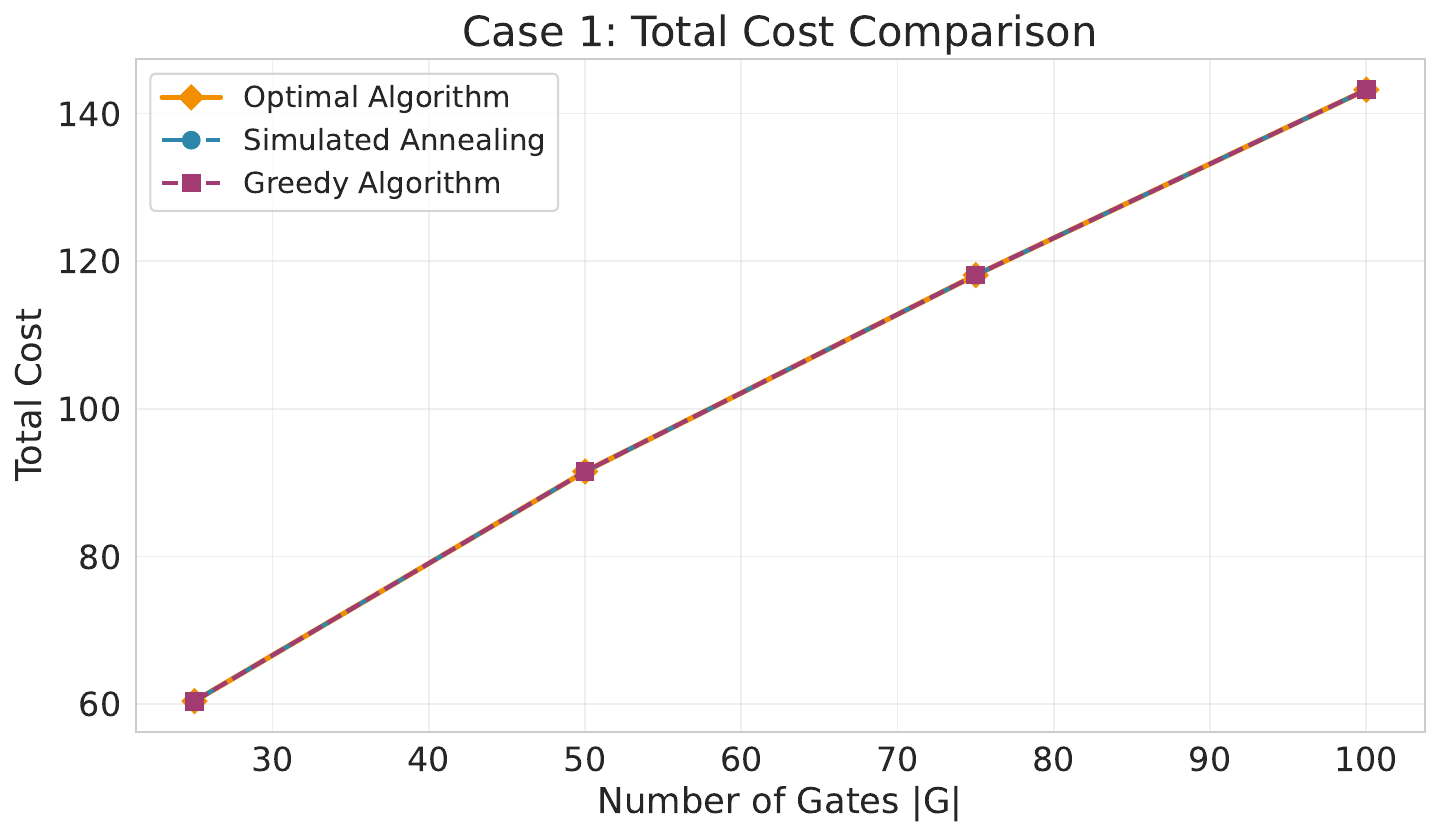}
\caption{The figure compares the total costs under the optimal, greedy, and 
SA algorithms for Case 1.}
\label{fig:case1_cost}
\end{figure}

\begin{figure}[htbp]
\centering
\includegraphics[width=0.4\textwidth]{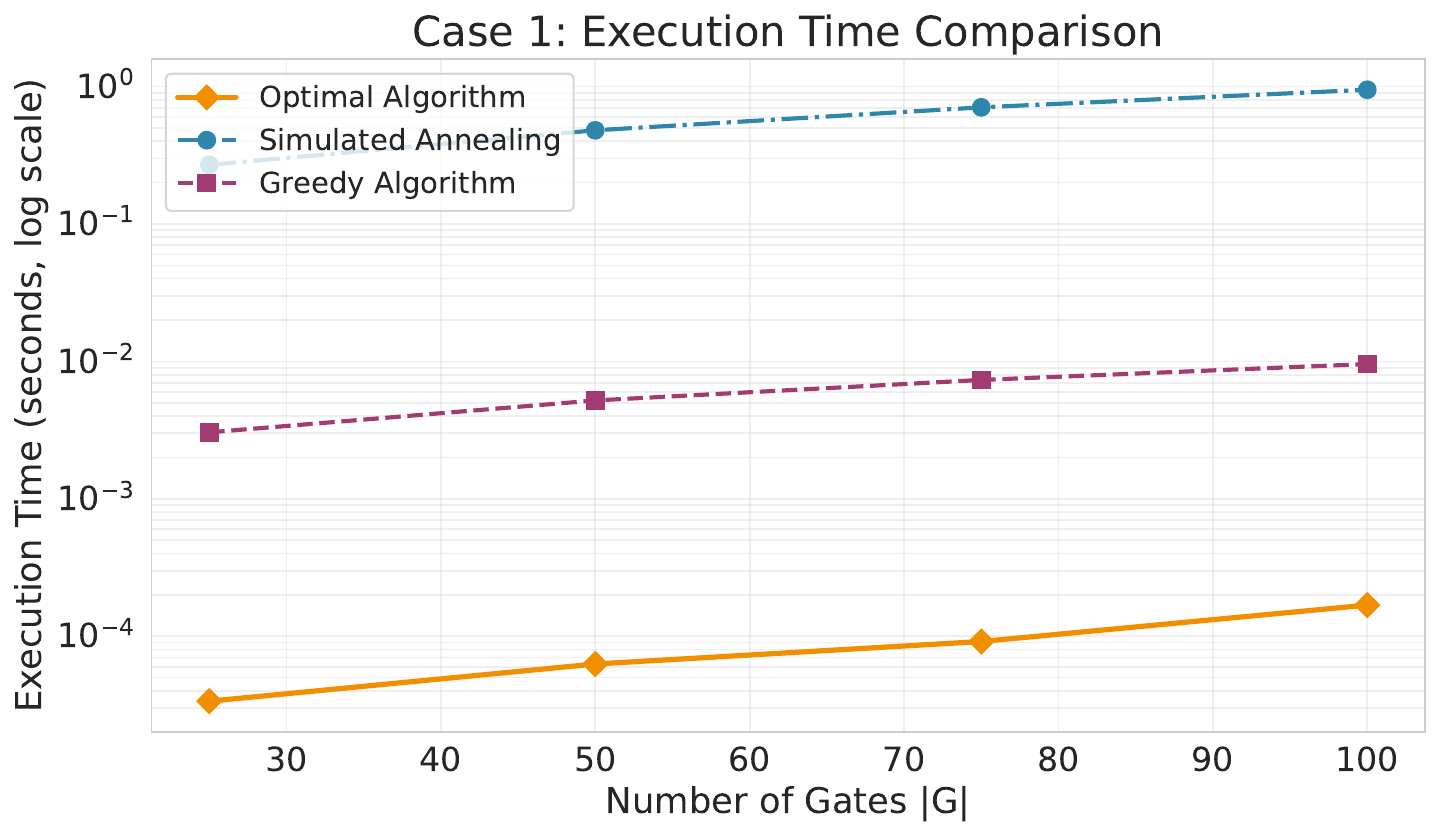}
\caption{The figure compares the execution times of the optimal, greedy, and 
SA algorithms for Case 1.}
\label{fig:case1_time}
\end{figure}

\subsubsection{Case 2 (Section \ref{subsec:homogeneous-zero-cost})}
For this special case, we set all QCs to have identical parameters: $s = 10$,
$e = 5$, $c^s = c^e = 2$, $c^g = 1$, $E(p', p) = 0$, and $A_{p,g} = 1$ for all
$p, p', g$, with $|P| = 5$, $|Q| = 15$ qubits, and $|G| = 50$ gates.
Fig. \ref{fig:case2_cost} compares the total costs for different values of
$\beta$. All three costs increase in $\beta$ as the algorithms prioritize
parallelization. The gap between the heuristics and the optimal solution is
largest at $\beta = 0$, where both heuristics are about $20-25 \, \%$ above
optimal, and it narrows as $\beta$ grows, to roughly $8 \, \%$ for the greedy
algorithm and under $5 \, \%$ for SA at $\beta = 2$.

This ordering deserves comment, since it is the opposite of what one might
expect. At $\beta = 0$ the objective rewards concentrating the circuit on as few
QCs as possible, which minimizes the peak execution occupancy $r^e_p$ that is
charged in \eqref{eq:obj}; the optimal algorithm achieves this by construction
through its exhaustive search over $k$, whereas both heuristics place each gate
where it is locally cheapest and thereby spread the circuit more widely than is
warranted. As $\beta$ grows, the makespan term rewards exactly the spreading
that the heuristics perform anyway, so their solutions move toward the optimal
one. Fig. \ref{fig:case2_time} shows that the optimal algorithm is an order of
magnitude faster than both heuristics, since for identical QCs it need only
evaluate $|P|$ candidate values of $k$.
\begin{figure}[htbp]
\centering
\includegraphics[width=0.4\textwidth]{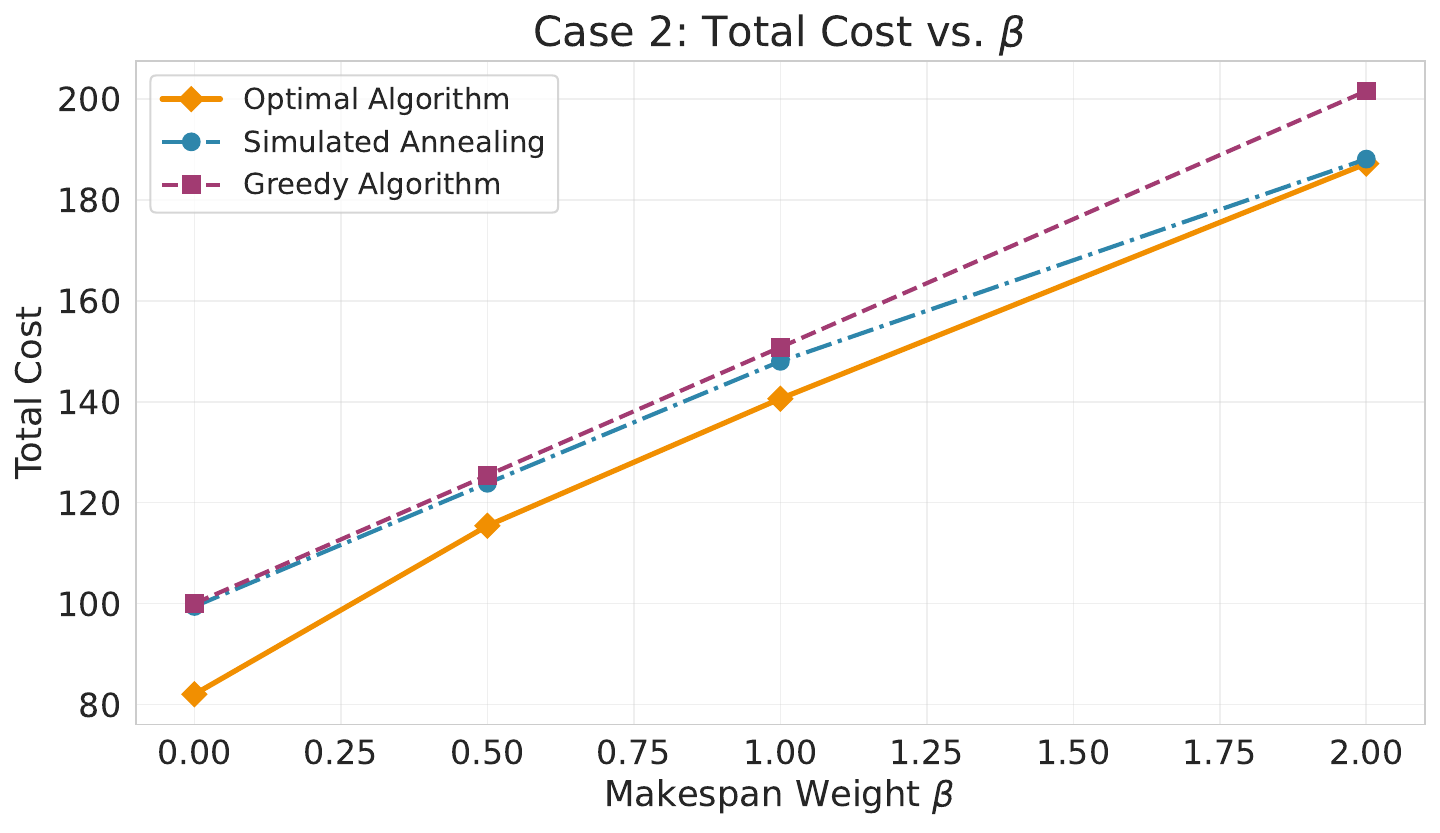}
\caption{The figure compares the total costs under the optimal, greedy, and 
SA algorithms for Case 2.}
\label{fig:case2_cost}
\end{figure}

\begin{figure}[htbp]
\centering
\includegraphics[width=0.4\textwidth]{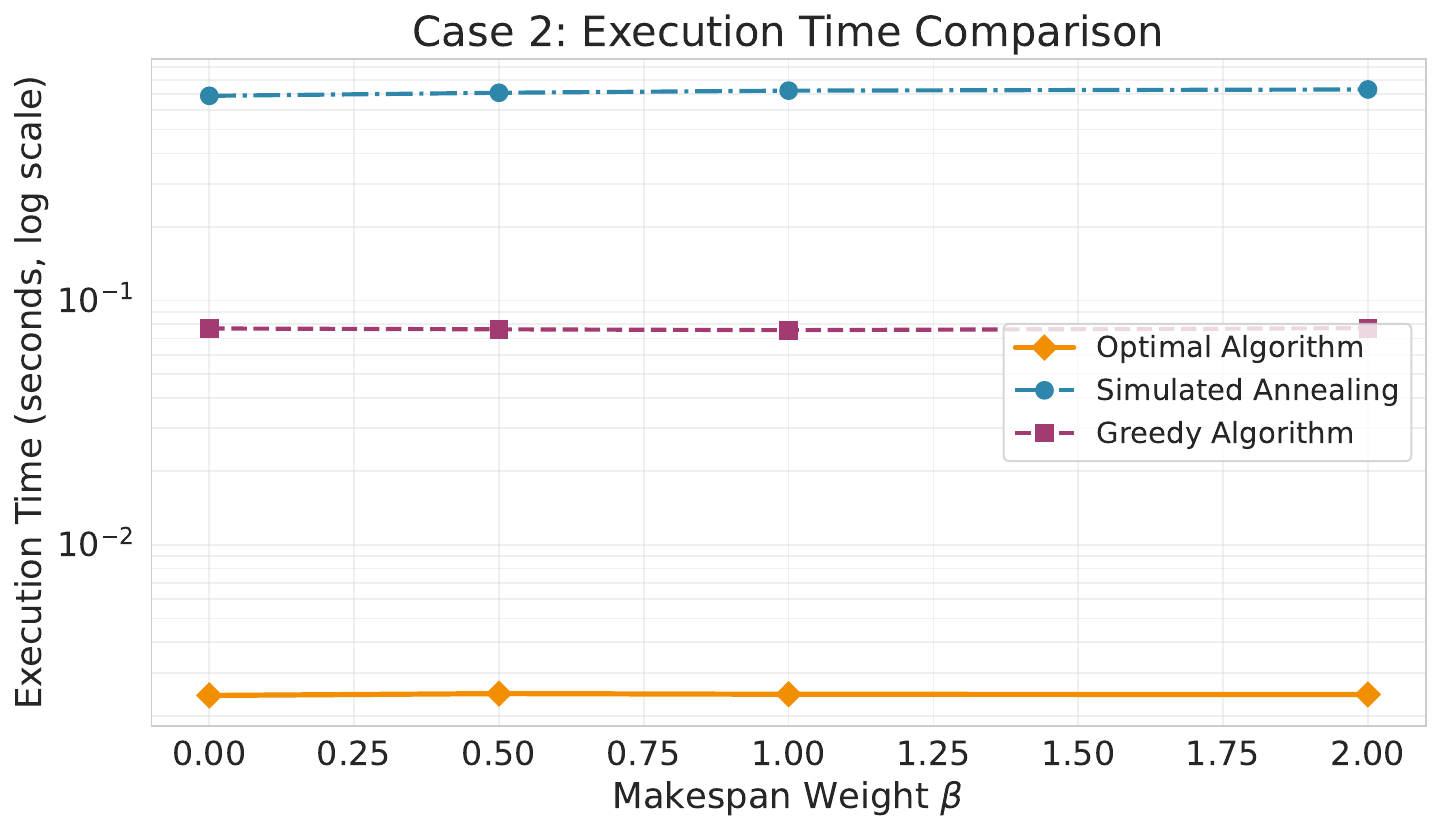}
\caption{The figure compares the execution times of the optimal, greedy, and 
SA algorithms for Case 2.}
\label{fig:case2_time}
\end{figure}

\subsubsection{Case 3 (Section \ref{subsec:chain-sequential})}
We generate sequential circuits with $|G|$ single-qubit gates on one qubit,
forming a linear chain, on a network of $|P| = 5$ QCs with heterogeneous costs.
Figs. \ref{fig:case3_cost} and \ref{fig:case3_time} compare the three
algorithms. As in Case 1, the three cost curves coincide across all circuit
sizes, growing linearly from about $40$ at $|G| = 20$ to about $145$ at
$|G| = 120$: the greedy algorithm attains the optimum found by Dijkstra's
algorithm on the time-expanded graph. The reason is structural. For a
single-qubit chain, the qubit occupies exactly one QC at any time and each gate
is placed where the sum of its execution cost and the teleportation cost from
the current QC is smallest, which is precisely the edge cost minimized along the
shortest path of Algorithm \ref{alg:singlequbit-path}. The greedy criterion is
therefore not an approximation of the optimal rule in this case, but coincides
with it. Fig. \ref{fig:case3_time} shows that the greedy algorithm is the
fastest of the three and its running time is nearly independent of $|G|$, while
the optimal algorithm's running time grows with the size of the time-expanded
graph and overtakes the greedy algorithm beyond $|G| \approx 40$. SA is the
slowest by two to three orders of magnitude.

\begin{figure}[htbp]
\centering
\includegraphics[width=0.4\textwidth]{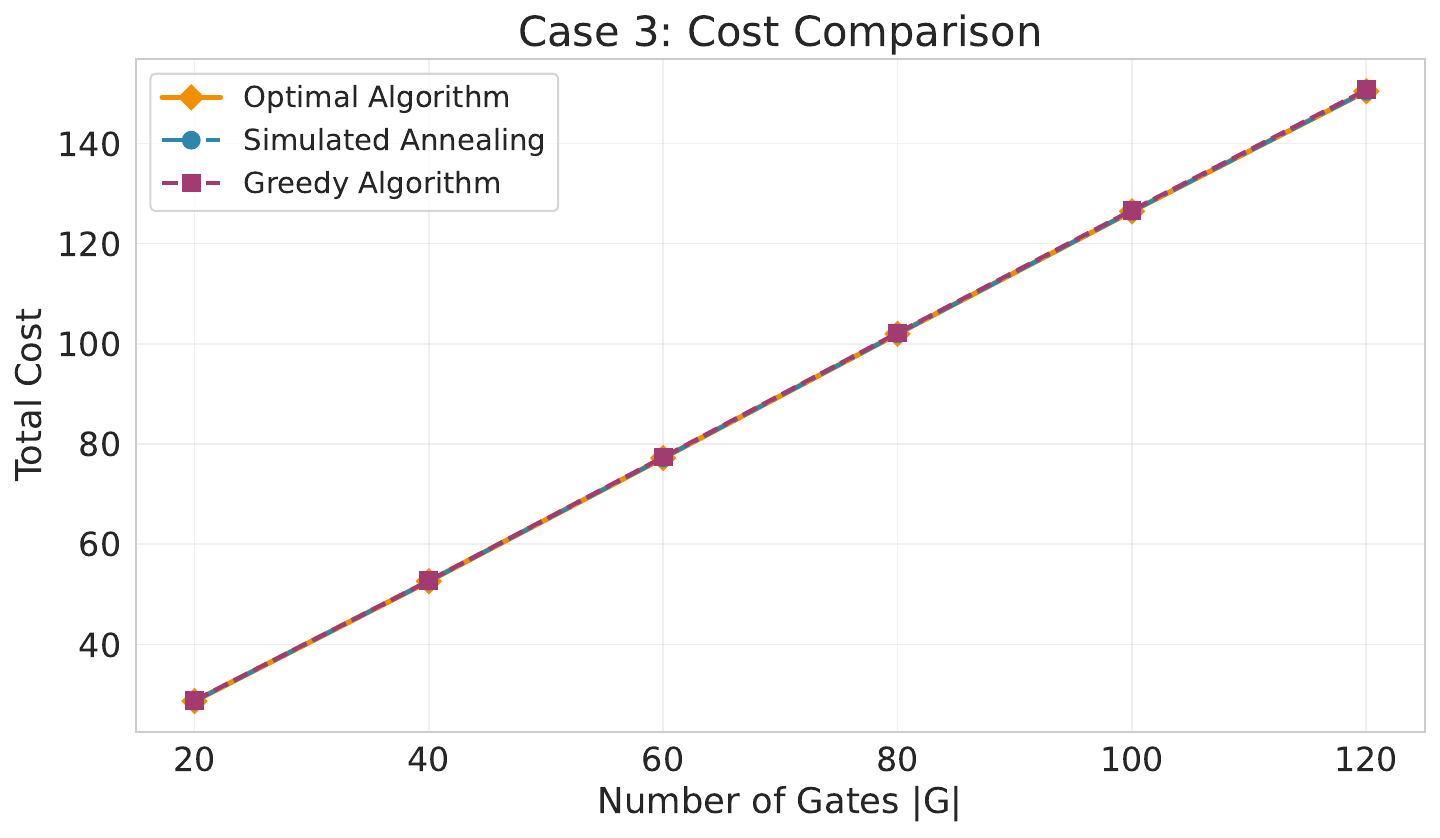}
\caption{The figure compares the total costs under the optimal, greedy, and 
SA algorithms for Case 3.}
\label{fig:case3_cost}
\end{figure}

\begin{figure}[htbp]
\centering
\includegraphics[width=0.4\textwidth]{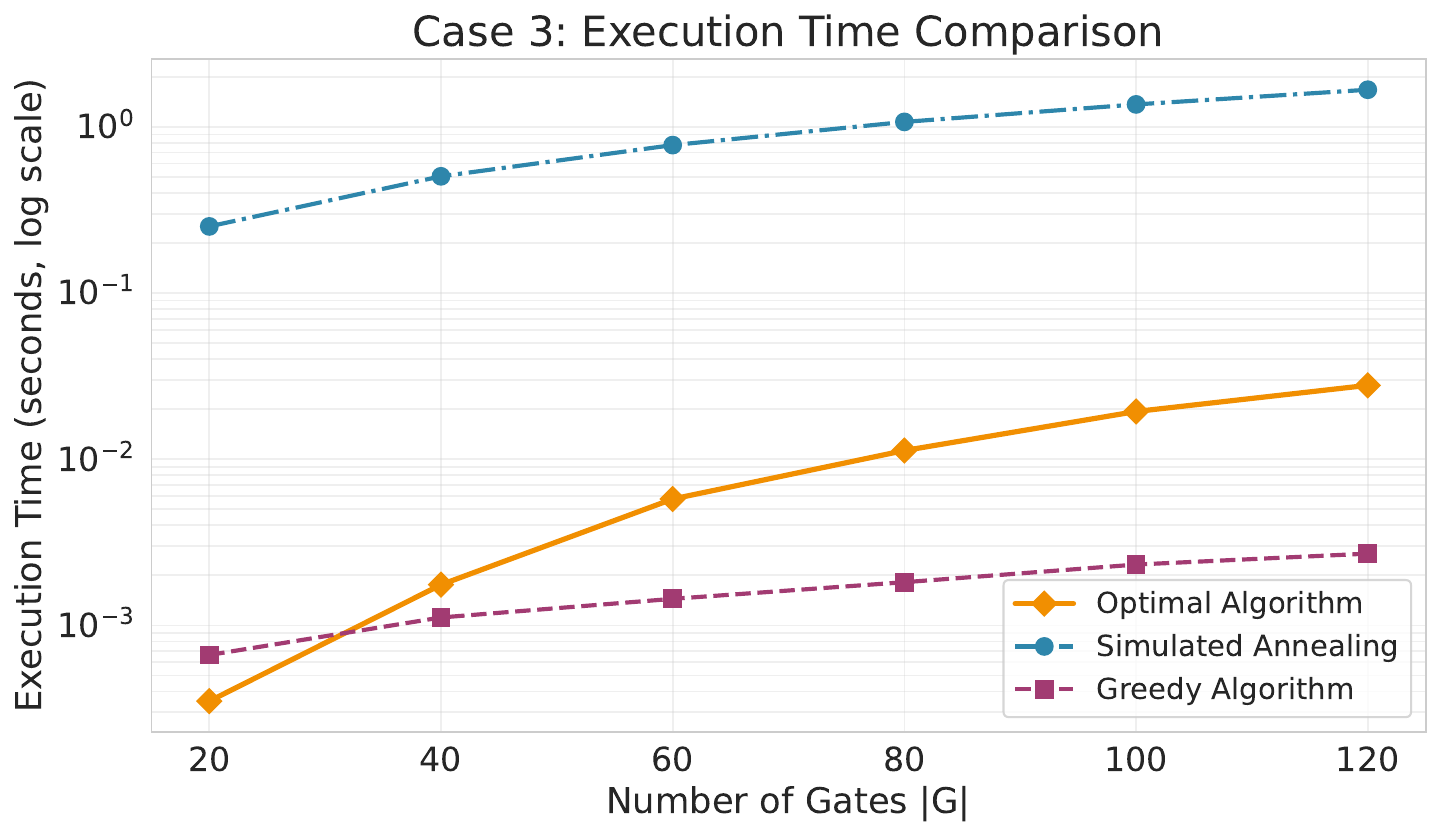}
\caption{The figure compares the execution times of the optimal, greedy, and 
SA algorithms for Case 3.}
\label{fig:case3_time}
\end{figure}

\subsection{Scalability Analysis}

To evaluate the scalability to large problem instances, we generated circuits 
with up to $|Q| = 30$ qubits, $|G| = 500$ gates, and $|P| = 8$ QCs. 
Fig. \ref{fig:scalability_time} shows the execution times versus the number 
of gates, $|G|$, on a log-log plot. The greedy algorithm's execution time 
grows approximately as $O(|G|^{1.15})$, close to theoretical complexity, 
while SA's execution time grows as $O(|G|^{1.9})$. For the largest instance 
($|Q| = 30$, $|G| = 500$, $|P| = 8$), the greedy algorithm completes in 
under $10$ seconds, while SA requires significantly longer. Memory usage 
remains manageable for the greedy algorithm even for the largest instances.

\begin{figure}[htbp]
\centering
\includegraphics[width=0.4\textwidth]{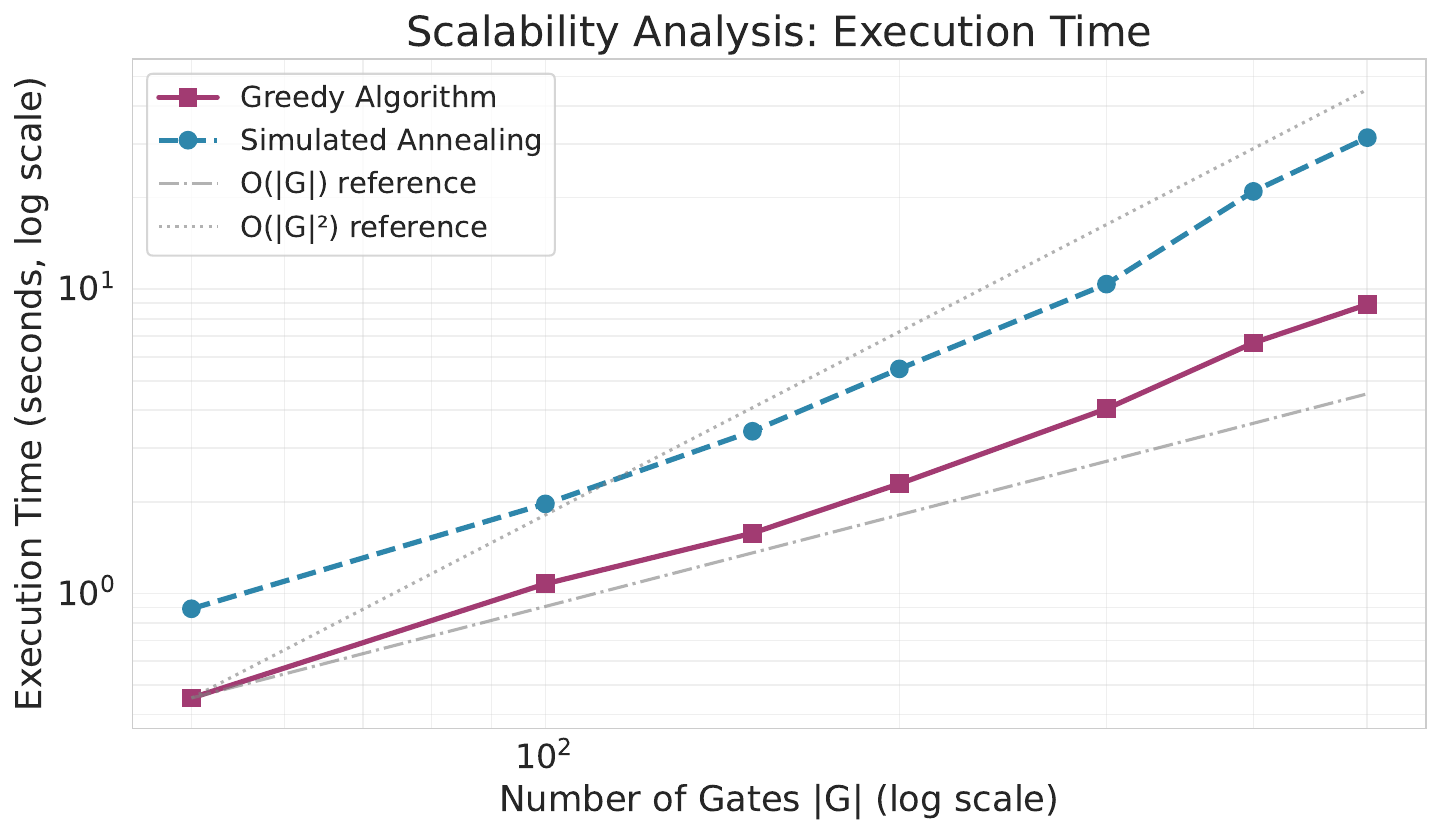}
\caption{The figure compares the scalability of the greedy and SA algorithms 
by showing their execution times on a log-log scale.}
\label{fig:scalability_time}
\end{figure}

Fig. \ref{fig:scalability_cost} shows the corresponding total costs. Both grow
approximately linearly with $|G|$. The greedy algorithm's cost exceeds SA's by
about $12 \, \%$ at $|G| = 100$, widening to roughly $19 \, \%$ at $|G| = 500$.
The gap grows slowly and sub-linearly, so the greedy algorithm's solution
quality degrades only gradually as instances scale, while its computational
advantage is retained.

\begin{figure}[htbp]
\centering
\includegraphics[width=0.4\textwidth]{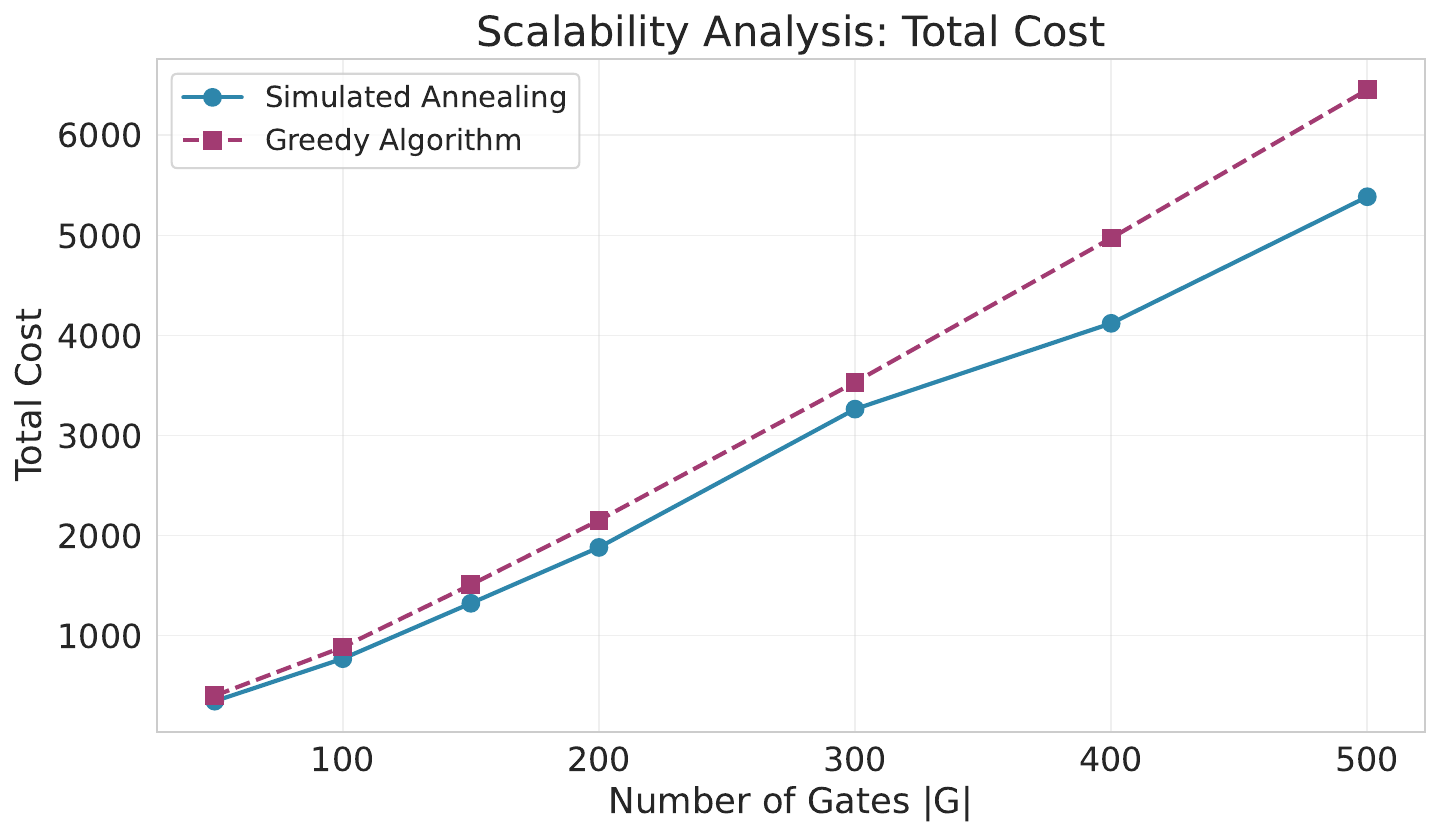}
\caption{The figure shows the total costs under the greedy and SA algorithms.}
\label{fig:scalability_cost}
\end{figure}

\subsection{Parameter Sensitivity Analysis}
We vary the weight parameters $\alpha_{\text{lease}}$ and
$\alpha_{\text{comm}}$ in the QC scoring function \eqref{eq:qc-score} to study
their impact on solution quality for three circuit types, with $|Q| = 20$,
$|G| = 100$, $|P| = 5$, and $\beta = 1.0$. The costs reported here are those of
the greedy QC selection and qubit allocation of
Algorithm \ref{alg:greedy-qc-selection} without the local search refinement of
Algorithm \ref{alg:local-search}, so that the figure isolates the effect of the
scoring weights; each curve is normalized within an instance before averaging.

Fig. \ref{fig:alpha_sensitivity} shows that all three circuit types favor a
low-to-moderate leasing weight, and that the cost rises monotonically once
$\alpha_{\text{lease}}$ exceeds about $0.5$. This is a consequence of the
instance family of Table \ref{table:instance_params}: teleportation costs are
drawn from Uniform$(10, 30)$ while leasing costs are drawn from
Uniform$(1, 5)$, so communication dominates the objective and a QC ranking that
ignores connectivity is penalized. The circuit types are separated by how
sharply they are penalized. The single-qubit dominated circuit ($60 \, \%$
single-qubit gates) is the most tolerant of a high leasing weight, attaining its
minimum near $\alpha_{\text{lease}} \approx 0.25$ and degrading by about
$10 \, \%$ at $\alpha_{\text{lease}} = 1$; with few two-qubit gates there is
little inter-qubit communication to mismanage. The two-qubit dominated circuit
($70 \, \%$ two-qubit gates) is the most sensitive, lying above the other two
curves over most of the range and degrading by roughly $10 \, \%$ from its
minimum, since every two-qubit gate whose operands are separated incurs a
teleportation. The balanced circuit lies between the two. Overall the cost
varies by about $10-13 \, \%$ across the parameter range, so the choice of
$\alpha_{\text{lease}}$ is worth tuning, and a value in $[0, 0.5]$ is a safe
default for communication-dominated instances.
\begin{figure}[htbp]
\centering
\includegraphics[width=0.4\textwidth]{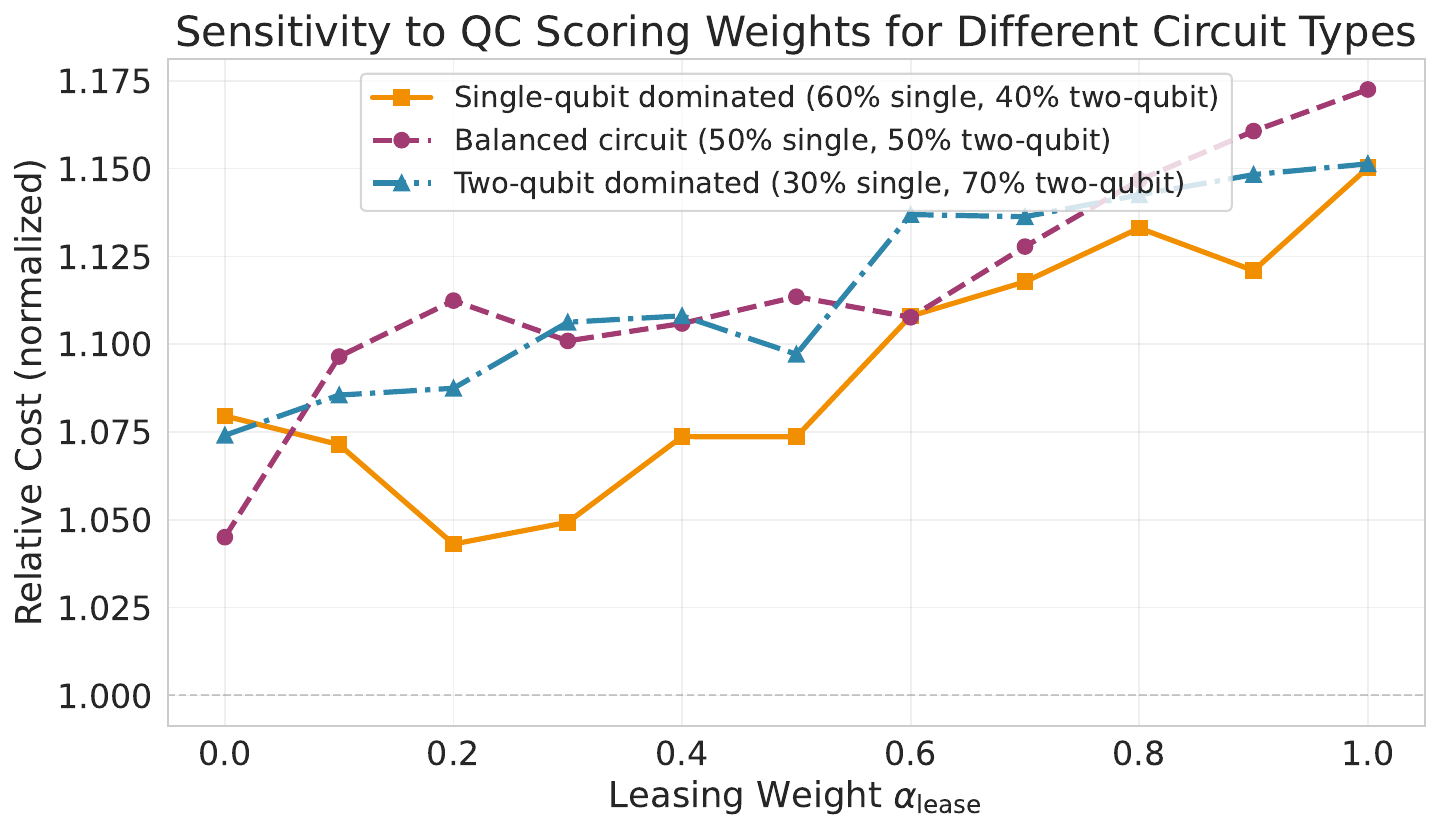}
\caption{The figure shows the sensitivity to QC scoring weights for different 
circuit types.}
\label{fig:alpha_sensitivity}
\end{figure}

\subsection{Summary of Numerical Results}
Table~\ref{tab:summary} summarizes the key performance metrics across all 
experiments. On general problem instances, the greedy algorithm produces 
solutions within $15 \, \%$ of SA, with the gap narrowing to a few percent for 
small circuits and for large numbers of QCs. In the special cases for which 
polynomial-time optimal algorithms are available, its behavior depends on the 
structure of the case: for Case 1 (unlimited-capacity QC) and Case 3 
(sequential circuits) it attains the optimal solution, whereas for Case 2 
(homogeneous QCs with zero teleportation cost) it is $8-25 \, \%$ above 
optimal. The two cases in which it is optimal are those in which the greedy 
placement rule coincides with the optimal rule-- centralized execution when 
teleportation out of $p_0$ cannot pay for itself, and cheapest-next-hop 
placement along a single-qubit chain-- while Case 2 is the case in which the 
objective rewards concentrating the circuit on few QCs, which a locally greedy 
gate placement does not do. The greedy algorithm is $8-40\times$ faster than 
SA, with the largest speedups on the instances with fewest QCs, and its 
execution time tracks the $O(|G|)$ reference closely, scaling to circuits with 
$500$ gates in seconds. Solution quality degrades only gradually with problem 
size: the gap to SA grows from about $12 \, \%$ at $|G| = 100$ to about 
$19 \, \%$ at $|G| = 500$. The choice of the scoring weights matters more than 
the default suggests-- cost varies by $10-13 \, \%$ across the range of 
$\alpha_{\text{lease}}$, and for the communication-dominated instances of 
Table~\ref{table:instance_params} a value in $[0, 0.5]$ is preferable to the 
balanced setting. Thus, for real-time distributed quantum computing scenarios, 
which require fast decisions, the greedy algorithm offers a good balance of 
solution quality and computational efficiency, and is optimal in structurally 
favorable cases.

\begin{table}[!ht]
\caption{The table shows a summary of the algorithm performance across all 
experiments.}
\label{tab:summary}
\centering
\scriptsize
\begin{tabular}{|l|r|r|r|}
\hline
\rowcolor[HTML]{EFEFEF}
\textbf{Metric} & \textbf{Greedy} & \textbf{SA} & \textbf{Optimal} \\ \hline
\multicolumn{4}{|c|}{\cellcolor[HTML]{D6EAF8}\textbf{General Instances}} \\ \hline
Avg. gap to SA & 2-15\% & --- & --- \\ \hline
Avg. exec. time & $< 1$ s & $\approx 10$ s & --- \\ \hline
Speedup over SA & $8-40\times$ & $1\times$ & --- \\ \hline
\multicolumn{4}{|c|}{\cellcolor[HTML]{D6EAF8}\textbf{Case 1: Single Unlimited QC}} \\ \hline
Avg. gap to opt. & $< 2$\% & $< 2$\% & 0\% \\ \hline
Avg. exec. time & $\approx 0.1$ s & $\approx 1$ s & $< 0.01$ s \\ \hline
\multicolumn{4}{|c|}{\cellcolor[HTML]{D6EAF8}\textbf{Case 2: Homogeneous, Zero Teleportation}} \\ \hline
Avg. gap to opt. & 8-25\% & 3-25\% & 0\% \\ \hline
Avg. exec. time & $\approx 0.02$ s & $\approx 0.1$ s & $\ll 0.01$ s \\ \hline
\multicolumn{4}{|c|}{\cellcolor[HTML]{D6EAF8}\textbf{Case 3: Sequential Gates}} \\ \hline
Avg. gap to opt. & $< 2$\% & $< 2$\% & 0\% \\ \hline
Avg. exec. time & $\approx 0.001$ s & $\approx 0.5$ s & $\approx 0.01$ s \\ \hline
\multicolumn{4}{|c|}{\cellcolor[HTML]{D6EAF8}\textbf{Large-Scale Instances}} \\ \hline
Max size tested & 30Q, 500G, 8P & 30Q, 500G, 8P & --- \\ \hline
\end{tabular}
\end{table}

\section{Conclusions and Future Work}
\label{sec:conclusion}
We presented a comprehensive ILP formulation for the JQLQCD problem and showed 
that it is NP-complete. Also, we identified several special cases in which the 
problem can be optimally solved in closed form or via polynomial-time 
algorithms, including those in which (A) there is an unlimited capacity QC in 
a heterogeneous network, (B) homogeneous QCs with zero teleportation cost, (C) 
chain topology with sequential gates, (D) independent subcircuits with 
partitioned resources, (E) infinite resources with makespan minimization only, 
and (F) a tree-structured circuit with an arbitrary QC network. We proposed a 
greedy algorithm with local search refinement for solving general instances of 
the JQLQCD problem. Using extensive numerical computations, we demonstrated 
that our proposed greedy algorithm achieves solutions within $15 \, \%$ of SA 
on general problem instances, while being $8-40\times$ faster, and that it 
attains the optimal solution for cases (A) and (C), in which the greedy 
placement rule coincides with the optimal one. For case (B), in which the 
objective rewards concentrating the circuit on few QCs, it is $8-25 \, \%$ 
above optimal. These results show that the greedy algorithm is practical for 
large-scale instances.

Some promising directions for future research are to develop approximation 
algorithms with provable guarantees for the general JQLQCD problem, 
investigate parameterized complexity for bounded structural parameters, and 
enhance the ILP with tight relaxations. Our numerical results suggest a more 
specific question along these lines: characterizing the structural conditions 
under which the greedy placement rule is exactly optimal, as it proved to be 
for cases (A) and (C), would delineate the instances on which no search is 
needed at all. A second question raised by our results is how to correct the 
tendency of a locally greedy gate placement to spread a circuit across more 
QCs than the objective warrants, which is the source of the remaining gap in 
case (B). Another open problem is to study extended formulations of the 
JQLQCD problem that incorporate fidelity tracking, dynamic circuit execution, 
multi-objective optimization, and stochastic programming. Another direction 
for future research is to design specialized strategies for solving the 
problem in the cases in which the quantum circuit to be run by the agent is a 
quantum variational algorithm circuit, QEC circuit, circuit for quantum 
simulation, etc. An important open problem is to experimentally validate the 
results of this paper via a real quantum network testbed. Finally, an 
interesting avenue for future research is to explore heterogeneous qubit 
technologies, error and fault tolerance, energy minimization, and dynamic 
pricing mechanisms for the JQLQCD problem.





\label{sec:conc}
\bibliographystyle{IEEEtran}
\bibliography{references}

\end{document}